\documentclass[11pt]{article}

\usepackage[a4paper,margin=1in]{geometry}
\usepackage[T1]{fontenc}
\usepackage[utf8]{inputenc}
\usepackage{lmodern}
\usepackage{microtype}
\usepackage{amsmath,amssymb,amsthm,mathtools}
\usepackage{bm}
\usepackage{booktabs}
\usepackage{array}
\usepackage{tabularx}
\usepackage{longtable}
\usepackage{enumitem}
\usepackage{xcolor}
\usepackage[sort,compress,noadjust]{cite}
\usepackage{algorithm}
\usepackage{algpseudocode}
\usepackage{caption}
\usepackage{makecell}
\usepackage{hyperref}

\definecolor{accent}{HTML}{1F4E79}
\definecolor{lightaccent}{HTML}{EAF2F8}

\newtheorem{theorem}{Theorem}
\newtheorem{lemma}{Lemma}

\newtheorem{fact}{Fact}
\newtheorem{corollary}{Corollary}
\newtheorem{definition}{Definition}
\newtheorem{example}{Example}
\newtheorem{remark}{Remark}

\newcommand{\F}{\mathbb F}
\newcommand{\Z}{\mathbb Z}
\newcommand{\C}{\mathbb{C}}

\newcommand{\ii}{\mathrm i}
\newcommand{\xor}{\oplus}
\newcommand{\ind}[1]{\mathbf 1_{\{#1\}}}
\newcommand{\lookup}[1]{\mathsf{#1}}

\newcommand{\ket}[1]{|#1\rangle}

\title{\sffamily Optimal dense materialization of the stabilizer formalism\\without polynomial overhead}
\author{
    Hyunho Cha \quad and \quad 
    Jungwoo Lee\thanks{Corresponding author: \texttt{junglee@snu.ac.kr}}\\
    \small NextQuantum Innovation Research Center\\
    \small Seoul National University, Seoul 08826, Republic of Korea
}
\date{}

\begin{document}
\maketitle

\begin{abstract}
Stabilizer states and Clifford transformations constitute the tractable backbone of quantum information science, from error correction and fault tolerance to benchmarking and simulation. Although these objects admit compact classical descriptions, many physical and computational workflows still require their explicit dense forms such as a full wavefunction for a stabilizer state or a full matrix for a Clifford transformation. In such explicit output tasks, exponential scaling is unavoidable because the outputs themselves have sizes $2^n$ and $4^n$. The fundamental question is therefore whether compact stabilizer and Clifford descriptions can be expanded with no additional polynomial overhead. Here we answer this question affirmatively. We present optimal algorithms that materialize an $n$-qubit stabilizer state vector in $O(2^n)$ time and a full dense Clifford matrix in $O(4^n)$ time. The same framework also yields an optimal conversion from standard stabilizer check matrices to state vectors and, for every fixed odd prime qudit dimension $\ell$, gives $O(\ell^n)$-time materialization of qudit stabilizer states. As an additional compact-to-compact result, we design a sign-aware Four Russians method for converting stabilizer check matrices to quadratic forms faster than Gaussian elimination. These results close the asymptotic gap between compact descriptions of the stabilizer formalism and their dense representations.
\end{abstract}

\section{Introduction}
\label{sec:intro}

A central difficulty in quantum many-body physics is that the natural
Hilbert space description of an \(n\)-qubit system has dimension \(2^n\),
whereas many physically important structures admit much smaller algebraic
descriptions. Stabilizer states and Clifford transformations are among the most
important tractable objects in this landscape. A stabilizer state is
the unique common \(+1\) eigenstate of an abelian group of Pauli constraints,
and therefore specifies a quantum state by its symmetries rather than by its
\(2^n\) complex amplitudes. This representation underlies
graph and cluster states in measurement-based quantum computation, code states and syndrome spaces in quantum error correction, and the Clifford dynamics captured by the Gottesman--Knill theorem \cite{shor1995scheme, steane1996error, cleve1997efficient, calderbank1998quantum, gottesman1998heisenberg, raussendorf2001one, schlingemann2001quantum, schlingemann2002stabilizer, hein2004multiparty, van2004graphical}. Modern stabilizer
simulators and tableau methods exploit precisely this structure to avoid
ever writing the full vector when the computation remains
inside the stabilizer sector~\cite{aaronson2004improved, anders2006fast, nest2010classical, gidney2021stim}.
\color{black}

The same compactness, however, creates a representation problem. An
\(n\)-qubit stabilizer state or Clifford unitary can be stored using only
polynomially many bits, for example as a check matrix, quadratic form or
tableau. Many calculations that surround stabilizer physics are nevertheless
formulated in the computational basis. Exact diagonalization routines,
dense linear algebra libraries, process matrix comparisons, simulator
validation, and studies of non-Clifford perturbations all expect
a wavefunction or a unitary matrix \cite{chuang1997prescription, poyatos1997complete, nielsen2002simple, gilchrist2005distance, greenbaum2015introduction, bravyi2016improved, bravyi2016trading, bravyi2019simulation, wietek2025xdiag}. In such settings, the goal is not efficient simulation of a Clifford circuit, but rather the conversion of a compact description into an explicit Hilbert-space object.

This boundary appears naturally in physical workflows where stabilizer
states serve as solvable reference points or calibration objects. A compact stabilizer pipeline may be the most
natural way to generate the object, while a dense vector or matrix may be
the required representation for comparing against a Schr\"odinger simulator,
computing basis-resolved diagnostics, interfacing with a non-stabilizer
routine, or producing benchmark data \cite{emerson2007symmetrized, magesan2011scalable, dankert2012exact, smelyanskiy2016qhipster, guerreschi2020intel, suzuki2021qulacs}. The conversion problem is
therefore not a replacement for stabilizer simulation. It is the interface
between the compact stabilizer formalism and the dense representations used
elsewhere in quantum physics.

For such explicit output problems, exponential scaling is unavoidable because
the requested physical object is itself exponentially large. A state vector
contains \(2^n\) amplitudes and a dense \(n\)-qubit unitary contains
\(4^n\) matrix entries. Thus the relevant question is not whether the
conversion can be polynomial time, but whether it can be output-size optimal:
can a compact stabilizer or Clifford description be expanded with no
additional factor polynomial in \(n\)?

This gap was visible in standard software pipelines, where comparable routes in libraries such as Qiskit \cite{javadi2024quantum} and Stim \cite{gidney2021stim} either did not expose the relevant materialization map directly or had worst-case costs with an \(O(n^2)\) overhead over the output size.
De Silva et al. showed that Gray code enumeration is a natural way to generate
the amplitudes of a stabilizer state from a quadratic form description
\cite{de2025fast}. That approach updates the support label in
constant time, but the phase update still scans a linear number of support
coordinates at each step. In the worst case this gives
\(O(n2^n)\) time for producing a \(2^n\)-entry state vector.
The analogous dense Clifford matrix conversion has target size \(4^n\), and the
same issue asks whether a compact tableau can be expanded at that
rate.

\begin{figure}
    \centering
    \includegraphics[width=\linewidth]{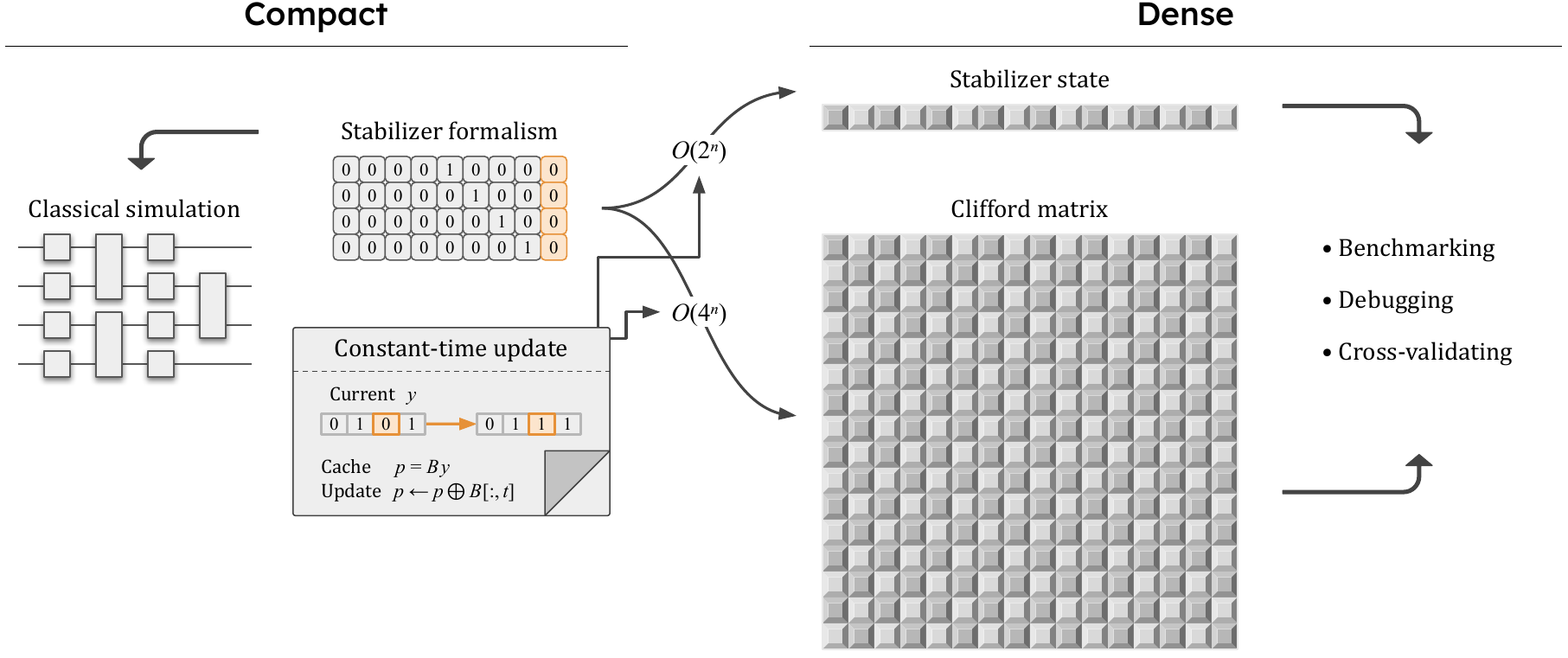}
    \caption{Overview of our dense materialization results. Compact descriptions are expanded into a dense stabilizer state vector in optimal \(O(2^n)\) time and into a dense Clifford matrix in optimal \(O(4^n)\) time. The essential part is the Gray code traversal update that caches the parity word.}
    \label{fig:schematic}
\end{figure}

Here we close this asymptotic gap. The underlying observation is that the
quadratic phase of a stabilizer state can be viewed as a binary interaction
graph on the coordinates of its affine support. During a Gray code
walk, only one support coordinate flips. If \(B\) denotes the off-diagonal
interaction matrix and \(y\) is the current Gray word, the word
\(p=By\) stores, simultaneously for every possible next flip, the parity of
the neighbors already active. The next phase increment is therefore
determined by one bit of \(y\) and one bit of \(p\). After the flip, all
future increments are refreshed by the single word operation
\(p\leftarrow p\oplus B[:,t]\). Thus the algorithm does not recompute a local
phase derivative at each step. It dynamically maintains all of them.

This invariant yields an \(O(2^n)\)-time algorithm for materializing a full
stabilizer state vector from a quadratic form description, matching the
number of amplitudes that must be written. Combined with the conversion from check matrices to quadratic forms, it also provides an output-optimal pipeline for converting check matrices to state vectors. We then extend the same principle to
Clifford gates. The first column of a Clifford matrix is a stabilizer state. The remaining columns can be traversed in Gray code order because adjacent columns differ by the Pauli image of a single-qubit
\(X\) operator. Applying one dense Pauli operator to one dense column costs
\(O(2^n)\), and \(2^n\) columns are produced, giving an \(O(4^n)\)-time
dense Clifford matrix expansion.
Figure~\ref{fig:schematic} illustrates our main dense materialization results.
The algorithm extends to every fixed odd prime qudit dimension $\ell$, and an $n$-qudit stabilizer state of local dimension $\ell$ can be materialized in $O(\ell^n)$ time.

In parallel, we also improve the polynomial preprocessing step that arises when the input is a stabilizer check matrix. In this case, the dense materialization pipeline first performs \(\text{check matrix}\to\text{quadratic form}\) and then applies the optimal dense expansion. Although the second step dominates any polynomial preprocessing in dense-output applications, the preprocessing cost is important in workflows
that keep stabilizer states compact. We design a \emph{sign-aware} variant of the Four Russians method \cite{arlazarov1970economical, gusfield1997algorithms, bard2006accelerating, albrecht2008efficient, bard2008matrix, albrecht2011efficient}, compatible with the sign convention of check matrices, that reduces the cost of this conversion from \(O(n^3)\) to \(O(n^3/\log n)\) bit operations.
Table~\ref{tab:complexity-comparison} summarizes the time complexity of the proposed algorithms.

\begin{table}[t]
\centering
\caption{Comparison of time complexities for stabilizer state and Clifford operator conversions.}
\label{tab:complexity-comparison}
\begin{tabular}{>{\raggedright\arraybackslash}p{0.42\linewidth} >{\raggedright\arraybackslash}p{0.28\linewidth} >{\raggedright\arraybackslash}p{0.20\linewidth}}
\toprule
Task & De Silva et al. \cite{de2025fast} & Ours \\
\midrule
Quadratic form $\to$ state vector & $O(n2^n)$ & $O(2^n)$ \\
Check matrix $\to$ state vector & $O(n2^n)$ & $O(2^n)$ \\
Clifford tableau $\to$ unitary matrix & $O(n4^n)$ & $O(4^n)$ \\
Check matrix $\to$ quadratic form & $O(n^3)$ & $O(n^3/\log n)$ \\
\bottomrule
\end{tabular}
\end{table}

\section{Preliminaries}

\subsection{Computational basis states and words}

We write $\F_2 = \{0,1\}$ for the field with two elements. Addition in $\F_2$ is XOR, denoted by $\xor$. If $x,y\in \F_2^n$, then $x\xor y$ means bitwise XOR.

Each $n$-bit string $x=(x_1,\dots,x_n)\in \F_2^n$ labels one computational basis state $\ket{x}$ of $n$ qubits. Thus an arbitrary state vector has the form
\[
\ket{\psi} = \sum_{x\in \F_2^n} \psi_x \ket{x},
\]
where the complex number $\psi_x$ is the amplitude at basis state $\ket{x}$.

When we say that an $n$-bit string ``fits in one word'', we mean the standard word-RAM convention in which the machine-word size is $\Theta(n)$ bits \cite{hagerup1998sorting}. This is the same convention under which \cite{de2025fast} explicitly counts bitwise XOR on $n$-bit labels as constant time in its complexity analysis (since it is parallelizable).
Our algorithms use only unit-cost word operations, including XOR, AND, shifts, comparison, integer addition modulo $4$, and array lookup on one-word indices.

\subsection{Stabilizer states}

The one-qubit Pauli matrices are
\[
I=\begin{bmatrix}1&0\\0&1\end{bmatrix},\quad
X=\begin{bmatrix}0&1\\1&0\end{bmatrix},\quad
Y=\begin{bmatrix}0&-\ii\\ \ii&0\end{bmatrix},\quad
Z=\begin{bmatrix}1&0\\0&-1\end{bmatrix}.
\]
An $n$-qubit Pauli operator is a tensor product of $n$ one-qubit Paulis, optionally multiplied by a sign. A stabilizer state is the unique common $+1$ eigenvector of $n$ independent commuting Hermitian Pauli operators.

\begin{definition}[Check matrix]
\label{def:check-matrix}
A stabilizer check matrix of an $n$-qubit stabilizer state is a pair
\[
(H,\sigma),
\qquad
H\in \F_2^{n\times 2n},
\qquad
\sigma\in \F_2^n,
\]
where the $r$-th row of $H$ is written as
\[
\bigl(w^{(r)} \mid u^{(r)}\bigr)\in \F_2^{2n}
\qquad (1\le r\le n),
\]
with $w^{(r)},u^{(r)}\in \F_2^n$. The first $n$ columns of $H$ store the $X$-patterns and the last $n$ columns store the $Z$-patterns of the generators. The row
$
\bigl(w^{(r)} \mid u^{(r)}\bigr)
$
together with the sign bit $\sigma_r$ encodes the Hermitian Pauli operator
\[
G_r
=
(-1)^{\sigma_r}
\bigotimes_{j=1}^n P\!\left(w^{(r)}_j,u^{(r)}_j\right),
\]
where
\[
P(0,0)=I,\qquad
P(1,0)=X,\qquad
P(0,1)=Z,\qquad
P(1,1)=Y.
\]
We require $G_1,\dots,G_n$ to be pairwise commuting and independent. The encoded stabilizer state is the unique state $\ket{\psi}$ satisfying
\[
G_r\ket{\psi}=\ket{\psi}
\qquad (r=1,\dots,n).
\]
\end{definition}

For our algorithm, the most useful fact is a standard representation theorem from the stabilizer literature due to \cite{dehaene2003clifford}, in the form used algorithmically by \cite{de2025fast}. Every stabilizer state can be written as a uniform-magnitude superposition over an affine subspace, with phases determined by a quadratic expression over bits.
We will take that representation as the input format.

\begin{definition}[Quadratic form]\label{def:qf}
Let $0\le k\le n$. A \emph{quadratic form description} of an $n$-qubit stabilizer state consists of:
\begin{itemize}[leftmargin=1.5em]
\item a shift vector $h\in \F_2^n$,
\item linearly independent basis vectors $v_1,\dots,v_k\in \F_2^n$,
\item a bit vector $d\in \F_2^k$,
\item an upper-triangular binary matrix $J\in \F_2^{k\times k}$,
\item a nonzero complex scalar $\gamma$.
\end{itemize}
These data specify amplitudes
\begin{equation}
\psi_{h \xor \sum_{t=1}^k y_t v_t} \;=\; \gamma \,\ii^{d^\top y}\,(-1)^{y^\top J y}
\qquad\textnormal{for every } y\in \F_2^k,
\label{eq:qf-description}
\end{equation}
and $\psi_x=0$ for every computational basis label $x\notin h+\operatorname{span}\{v_1,\dots,v_k\}$.

Here all products of bits are ordinary integer products, and the exponents are interpreted modulo $4$ in the obvious way:
\[
\ii^0=1,\qquad \ii^1=\ii,\qquad \ii^2=-1,\qquad \ii^3=-\ii.
\]
\end{definition}

The affine subspace
\[
A \;=\; h + \operatorname{span}\{v_1,\dots,v_k\}
\]
is the \emph{support} of the state, and $k=\dim A$ is its support dimension.

\subsection{The concrete algorithmic problem}

We consider two explicit output materialization tasks.

For stabilizer states, the input is a quadratic form description as in
Definition~\ref{def:qf}, and the output is the \emph{full} $2^n$-entry
amplitude vector
\[
(\psi_x)_{x\in \F_2^n}.
\]

For Clifford gates, the input is a compact Clifford tableau for an $n$-qubit
Clifford operator $C$, namely the Pauli images
\begin{equation}
\label{eq:Ut_and_Vt_definition}
U_t = C Z_t C^\dagger,
\qquad
V_t = C X_t C^\dagger
\qquad
(1\le t\le n),
\end{equation}
where $Z_t$ and $X_t$ denote the single-qubit Pauli operators acting on qubit
$t$ \cite{koenig2014efficiently, selinger2015generators, bravyi2021hadamard}. The output is the \emph{full} $2^n\times 2^n$ dense matrix
\[
(C_{z,x})_{z,x\in \F_2^n},
\qquad
C_{z,x}=\langle z|C|x\rangle,
\]
equivalently, the list of all dense columns $c_x=C|x\rangle$.

\begin{remark}
Throughout the paper, compact stabilizer and Clifford descriptions are understood projectively
unless an explicit scalar phase is supplied. Thus a stabilizer check matrix specifies the
one-dimensional common $+1$ eigenspace, rather than a distinguished normalized vector, and a Clifford tableau specifies the conjugation action of a Clifford operator only up to a global phase. When we speak of materializing a dense state vector or Clifford matrix, we mean materializing a chosen representative of this projective object.
\end{remark}

This explicit output setting is different from the usual polynomial-time simulation problem. In ordinary simulation one keeps a
compact description throughout. Here, by contrast, the goal is specifically to
\emph{materialize} the dense state vector or the dense Clifford matrix.

De Silva et al.~\cite{de2025fast} showed that Gray code is the natural traversal
order for the state vector problem. Their algorithm updates the current support
index in constant time, but its phase update still costs linear time in the
support dimension $k$. Hence the complexity is $O(n2^n)$ in the worst case
$k=n$. For the Clifford matrix problem, the analogous dense reconstruction from
a compact tableau has output size $4^n$ but still carries an extra linear
overhead, giving $O(n4^n)$. Our goal is to remove these remaining factors of
$n$, obtaining $O(2^n)$ time for the full stabilizer state vector and $O(4^n)$
time for the full Clifford matrix.

\section{Gray code and the parity word invariant}

\subsection{Gray code properties}

We will use the binary reflected Gray code
\[
g(m)=m \xor \bigl\lfloor m/2 \bigr\rfloor,\qquad m=0,1,\dots,2^k-1,
\]
where the integer $g(m)$ is identified with its $k$-bit binary expansion \cite{gilbert1958gray, savage1997survey}.

\begin{fact}
\label{fact:gray}
Fix $k\ge 1$. Define $g(m)=m \xor \lfloor m/2\rfloor$ for $m\in\{0,\dots,2^k-1\}$. Then:
\begin{enumerate}[label=\textup{(\roman*)}, leftmargin=1.8em]
\item the map $m\mapsto g(m)$ is a bijection from $\{0,\dots,2^k-1\}$ onto $\F_2^k$;
\item for every $m\in\{1,\dots,2^k-1\}$,
\[
g(m)\xor g(m-1) \;=\; m \mathbin{\&} (-m),
\]
where $\&$ is bitwise AND. In particular, consecutive Gray words differ in exactly one bit.
\end{enumerate}
\end{fact}

\subsection{Rewriting the phase}

The phase in Eq.~\eqref{eq:qf-description} has a linear part and a quadratic part. It is convenient to separate the diagonal and off-diagonal contributions.

\begin{definition}[Linear coefficients and interaction matrix]
\label{def:interaction}
Let $d\in \F_2^k$ and let $J\in \F_2^{k\times k}$ be upper triangular. Define:
\[
a_t = d_t + 2J_{tt}\in \Z_4 \qquad (1\le t\le k),
\]
and define the symmetric zero-diagonal matrix $B\in \F_2^{k\times k}$ by
\[
B_{ij}=
\begin{cases}
J_{ij}, & i<j,\\
0, & i=j,\\
J_{ji}, & i>j.
\end{cases}
\]
Then for every $y\in \F_2^k$,
\begin{equation}
d^\top y + 2\, y^\top J y
\equiv
\sum_{t=1}^k a_t y_t
+
2\sum_{1\le i<j\le k} B_{ij} y_i y_j
\pmod 4.
\label{eq:phase-decomp}
\end{equation}
\end{definition}

We now introduce the cached quantity that makes Gray-step updates faster.

\begin{definition}[Parity vector]
For $y\in \F_2^k$, define the \emph{parity vector}
\[
p(y)=By\in \F_2^k.
\]
Equivalently, its $t$-th component is
\[
p_t(y)=\sum_{j=1}^k B_{tj} y_j \pmod 2.
\]
Because $B$ has zero diagonal, $p_t(y)$ depends on every active support bit \emph{except} $y_t$ itself.
\end{definition}

The entire algorithm rests on the next two lemmas.

\begin{lemma}[Phase increment of a single Gray code flip]
\label{lem:phase-increment}
Let
\[
\phi(y)=\sum_{t=1}^k a_t y_t + 2\sum_{1\le i<j\le k} B_{ij} y_i y_j \pmod 4.
\]
Then for every $y\in \F_2^k$ and every standard basis vector $e_t\in \F_2^k$,
\[
\phi(y\xor e_t)-\phi(y)\equiv (1-2y_t)a_t + 2\,p_t(y)\pmod 4.
\]
Equivalently,
\[
\phi(y\xor e_t)-\phi(y)\equiv
\begin{cases}
+a_t + 2p_t(y)\pmod 4, & y_t=0,\\
-a_t + 2p_t(y)\pmod 4, & y_t=1.
\end{cases}
\]
\end{lemma}

\begin{proof}
Let $y'=y\xor e_t$. Then $y'_t=1-y_t$ and $y'_j=y_j$ for $j\ne t$.

For the linear part,
\[
\sum_{s=1}^k a_s y'_s-\sum_{s=1}^k a_s y_s
=
a_t(y'_t-y_t)
=
a_t\bigl((1-y_t)-y_t\bigr)
=
(1-2y_t)a_t.
\]

Now consider the quadratic part
\[
Q(y)=2\sum_{1\le i<j\le k} B_{ij} y_i y_j.
\]
Only monomials containing $y_t$ can change, so
\[
Q(y')-Q(y)
=
2\sum_{j<t} B_{jt}\bigl(y_j y'_t-y_j y_t\bigr)
+
2\sum_{t<j} B_{tj}\bigl(y'_t y_j-y_t y_j\bigr).
\]
Since $y'_t-y_t=(1-y_t)-y_t=1-2y_t$, this becomes
\[
Q(y')-Q(y)
=
2(1-2y_t)\left(\sum_{j<t} B_{jt} y_j + \sum_{t<j} B_{tj} y_j\right).
\]
Because $1-2y_t\in\{1,-1\}$ and $2(-1)\equiv 2 \pmod 4$, we have
\[
2(1-2y_t)\equiv 2 \pmod 4.
\]
Hence
\[
Q(y')-Q(y)
\equiv
2\left(\sum_{j<t} B_{jt} y_j + \sum_{t<j} B_{tj} y_j\right)
\pmod 4.
\]
Since $B$ is symmetric and has zero diagonal, the right-hand side is
\[
2\sum_{j\ne t} B_{tj} y_j = 2\,p_t(y)
\pmod 4.
\]

Adding the linear and quadratic differences gives
\[
\phi(y')-\phi(y)\equiv (1-2y_t)a_t + 2p_t(y)\pmod 4.
\]
\end{proof}

\begin{lemma}
\label{lem:parity-update}
Let $c_t\in \F_2^k$ be the $t$-th column of $B$. Then for every $y\in\F_2^k$,
\[
p(y\xor e_t)=p(y)\xor c_t.
\]
\end{lemma}

\begin{proof}
Because $p(y)=By$ over $\F_2$,
\[
p(y\xor e_t)=B(y\xor e_t)=By \xor Be_t = p(y)\xor c_t.
\]
\end{proof}

\begin{remark}
Lemma~\ref{lem:phase-increment} says that if we already know the current bit $y_t$ and the parity bit $p_t(y)$, then the next phase exponent needs only constant work. Lemma~\ref{lem:parity-update} says that after the bit flip, we can update \emph{all} future phase parities at once by one XOR with a precomputed column of $B$. Thus the algorithm keeps two compact invariants during the Gray walk, namely the current Gray word $y$ itself and the parity word $By$. We cache the whole vector $By$ instead of recomputing a single entry from scratch at every step.
\end{remark}

\section{Optimal materialization of stabilizer states}
\label{sec:stabilizer-materialization}

\subsection{A word-oriented implementation}

We encode each $k$-bit vector as one machine word. During the Gray code traversal we maintain the current Gray word $y\in\F_2^k$ as well as the parity word $p(y)=By\in\F_2^k$. For the Gray code traversal we use the one-hot flip word
\[
f_m = m \mathbin{\&} (-m), \qquad m=1,2,\dots,2^k-1.
\]
By Fact~\ref{fact:gray}, this is the single bit that changes between Gray words $g(m-1)$ and $g(m)$.

To avoid any hidden primitive for extracting the bit index from the one-hot word $f_m$, we use one-hot lookup tables. Concretely, we allocate arrays indexed by integers in $\{0,\dots,2^k-1\}$ and fill only the power-of-two locations:
\[
\lookup{A}[2^{t-1}] = a_t,\qquad
\lookup{V}[2^{t-1}] = v_t,\qquad
\lookup{C}[2^{t-1}] = c_t.
\]
This keeps the algorithm completely elementary. Since $2^k\le 2^n$, the extra table size does not affect the final asymptotic bound.

A subtle but important point is that the linear phase term changes by $+a_t$ when bit $t$ flips from $0$ to $1$, but by $-a_t$ when it flips back from $1$ to $0$. That is why the algorithm stores the current Gray word $y$ in addition to the parity word $By$.

\begin{algorithm}[t]
\caption{}
\label{alg:main}
\begin{algorithmic}[1]
\Require $n$, $k$, shift $h\in\F_2^n$, linearly independent $v_1,\dots,v_k\in\F_2^n$, phase data $d\in\F_2^k$, upper-triangular $J\in\F_2^{k\times k}$, scalar $\gamma\in\mathbb C\setminus\{0\}$
\Ensure Full array $(\psi_x)_{x\in\F_2^n}$

\State Build $a_1,\dots,a_k$ and $B$ from Definition~\ref{def:interaction}
\State Initialize one-hot lookup tables $\lookup{A},\lookup{V},\lookup{C}$ of length $2^k$ to zero
\For{$t=1$ to $k$}
    \State $\lookup{A}[2^{t-1}] \gets a_t$
    \State $\lookup{V}[2^{t-1}] \gets v_t$
    \State $\lookup{C}[2^{t-1}] \gets$ the $t$-th column of $B$
\EndFor
\State Initialize a length-$2^n$ complex array $\psi$ to all zeros
\State $y \gets 0$ \Comment{current Gray word in $\F_2^k$}
\State $x \gets h$ \Comment{current support index in $\F_2^n$}
\State $p \gets 0$ \Comment{current parity word $p(y)=By$ in $\F_2^k$}
\State $q \gets 0$ \Comment{current phase exponent in $\Z_4$}
\State $\psi_x \gets \gamma$
\For{$m=1$ to $2^k-1$}
    \State $f \gets m \mathbin{\&} (-m)$ \Comment{one-hot Gray code flip word}
    \If{$(y \mathbin{\&} f)=0$}
        \State $q \gets q + \lookup{A}[f] + 2\,\ind{(p \mathbin{\&} f)\neq 0}\pmod 4$
    \Else
        \State $q \gets q - \lookup{A}[f] + 2\,\ind{(p \mathbin{\&} f)\neq 0}\pmod 4$
    \EndIf
    \State $y \gets y \xor f$
    \State $x \gets x \xor \lookup{V}[f]$
    \State $p \gets p \xor \lookup{C}[f]$
    \State $\psi_x \gets \gamma\,\ii^q$
\EndFor
\State \Return $\psi$
\end{algorithmic}
\end{algorithm}

\subsection{Correctness}

We now prove that Algorithm~\ref{alg:main} writes the amplitudes prescribed by Definition~\ref{def:qf} (equivalently, by Eq.~\eqref{eq:qf-description}).

\begin{theorem}
\label{thm:correctness}
Algorithm~\ref{alg:main} returns the full amplitude vector of the stabilizer state specified by the input quadratic form description.
\end{theorem}

\begin{proof}
Let $g(0),g(1),\dots,g(2^k-1)$ be the binary reflected Gray code from Fact~\ref{fact:gray}. For each $m$, define
\[
y^{(m)} = g(m)\in \F_2^k,
\qquad
x^{(m)} = h \xor \sum_{t=1}^k y^{(m)}_t v_t \in \F_2^n,
\]
and let
\[
q^{(m)} = \phi\bigl(y^{(m)}\bigr)\in \Z_4,
\qquad
p^{(m)} = p\bigl(y^{(m)}\bigr)=B y^{(m)}\in \F_2^k.
\]
We prove by induction on $m$ that after the $m$-th write of the algorithm (counting the initial write as $m=0$), the internal variables satisfy
\[
y = y^{(m)},\qquad x = x^{(m)},\qquad p=p^{(m)},\qquad q=q^{(m)},
\]
and the array entry written at that moment is
\[
\psi_{x^{(m)}} = \gamma \,\ii^{q^{(m)}}.
\]

\textbf{Base case $m=0$.} The algorithm sets $y=0$, $x=h$, $p=0$, and $q=0$, then writes $\psi_h=\gamma$. Since $y^{(0)}=g(0)=0$, we have
\[
y^{(0)}=0,\qquad x^{(0)}=h,\qquad p^{(0)}=B0=0,\qquad q^{(0)}=\phi(0)=0.
\]
So the claim holds at $m=0$.

\textbf{Induction step.} Assume the claim holds after the $(m-1)$-st write, where $1\le m\le 2^k-1$. By Fact~\ref{fact:gray},
\[
y^{(m)} = y^{(m-1)} \xor e_t
\]
for the unique coordinate $t$ whose one-hot word is
\[
f_m = g(m)\xor g(m-1) = m\mathbin{\&}(-m).
\]

Because the current algorithmic variable is $y=y^{(m-1)}$ by the induction hypothesis, the condition $(y\& f_m)=0$ is equivalent to $y_t^{(m-1)}=0$, and $(y\& f_m)\neq 0$ is equivalent to $y_t^{(m-1)}=1$.

First the algorithm updates the phase exponent. Since $(p^{(m-1)} \mathbin{\&} f_m)\neq 0$ exactly when the $t$-th bit of $p^{(m-1)}$ is $1$, Algorithm~\ref{alg:main} performs the update
\[
q \gets
\begin{cases}
q + a_t + 2\,p_t\bigl(y^{(m-1)}\bigr)\pmod 4, & y_t^{(m-1)}=0,\\[0.4em]
q - a_t + 2\,p_t\bigl(y^{(m-1)}\bigr)\pmod 4, & y_t^{(m-1)}=1.
\end{cases}
\]
By the two-case form of Lemma~\ref{lem:phase-increment}, this is
\[
q \gets q + \phi\bigl(y^{(m)}\bigr)-\phi\bigl(y^{(m-1)}\bigr)\pmod 4.
\]
Using the induction hypothesis $q=q^{(m-1)}=\phi(y^{(m-1)})$, we conclude that the updated variable satisfies
\[
q = q^{(m)}.
\]

Next the algorithm updates the Gray word:
\[
y \gets y \xor f_m = y^{(m-1)} \xor e_t = y^{(m)}.
\]

Because $\lookup{V}[f_m]=v_t$, the algorithm updates
\[
x \gets x \xor \lookup{V}[f_m] = x^{(m-1)} \xor v_t = x^{(m)}.
\]

Because $\lookup{C}[f_m]=c_t$, where $c_t$ is the $t$-th column of $B$, Lemma~\ref{lem:parity-update} gives
\[
p \gets p \xor \lookup{C}[f_m]
=
p^{(m-1)} \xor c_t
=
p\bigl(y^{(m-1)}\xor e_t\bigr)
=
p^{(m)}.
\]

Hence the new write is
\[
\psi_{x^{(m)}} = \gamma \,\ii^{q^{(m)}}.
\]
By Definition~\ref{def:qf}, that is precisely the correct amplitude at the support point indexed by $x^{(m)}$.

It remains to show that every support point is written exactly once and every non-support point remains zero. Fact~\ref{fact:gray}(i) says that the Gray code visits every $y\in\F_2^k$ exactly once. Since the vectors $v_1,\dots,v_k$ are linearly independent, the map
\[
y \longmapsto h \xor \sum_{t=1}^k y_t v_t
\]
is injective. Indeed, if
\[
h \xor \sum_{t=1}^k y_t v_t = h \xor \sum_{t=1}^k z_t v_t,
\]
then XORing both sides with $h$ yields
\[
\sum_{t=1}^k (y_t\xor z_t)\,v_t = 0.
\]
Linear independence of the $v_t$ forces $y_t\xor z_t=0$ for every $t$, hence $y=z$. So the support points $x^{(m)}$ are all distinct, and every element of the affine support is written exactly once. All other array entries were initialized to zero and never changed. Therefore the returned array is the full amplitude vector.
\end{proof}

\begin{figure}
    \centering
    \includegraphics[width=0.9\linewidth]{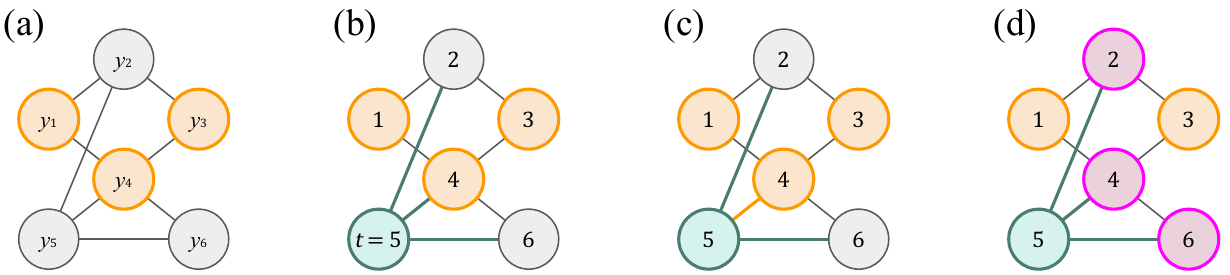}
    \caption{A $k=6$ example of a quadratic phase represented as a graph. Vertices and edges represent the support coordinates $y_t$ and off-diagonal terms $B_{ij}y_iy_j$, respectively. (a) For the active support word $y = (1, 0, 1, 1, 0, 0)$, the resulting active set is $\{1, 3, 4\}$. (b) Before the flip, node 5 possesses only one active neighbor (node 4) out of its neighborhood $\{2, 4, 6\}$, leading to a parity value of $p_5(y) = 1$. (c) Following the flip to the new active set $\{1, 3, 4, 5\}$, the odd parity of active neighbors for vertex 5 yields a sign contribution of $\ii^{2p_5} = -1$ via the edge $(4, 5)$. (d) Since column $c_5 = B[:,5]$ contains ones exactly at positions $\{2, 4, 6\}$, flipping vertex 5 toggles the parity values cached at those neighboring vertices.}
    \label{fig:graph-example}
\end{figure}

A graphical example is illustrated in Figure~\ref{fig:graph-example}.

\subsection{Complexity and optimality}

\begin{theorem}
\label{thm:complexity}
On the word-RAM model described above, Algorithm~\ref{alg:main} runs in time
$
O(2^n)
$
and uses space
$
O(2^n).
$
More precisely, its running time is
$
O(2^n + 2^k + k^2),
$
and its working memory beyond the output array is
$
O(2^k + k^2)
$
words.
\end{theorem}

\begin{proof}
There are three cost components.

\begin{enumerate}
\item \textbf{Building $a$ and $B$.} Computing the $k$ values $a_t$ and the $k\times k$ interaction matrix $B$ takes $O(k^2)$ time.

\item \textbf{One-hot lookup tables.} Initializing and filling the one-hot tables $\lookup{A},\lookup{V},\lookup{C}$ takes $O(2^k + k)$ time and $O(2^k)$ words of space.

\item \textbf{Output array and Gray traversal.} Initializing the output array of length $2^n$ to zero costs $O(2^n)$ time and $O(2^n)$ space. The Gray loop has $2^k-1$ iterations. In each iteration the algorithm performs:
\begin{itemize}[leftmargin=1.5em]
\item one least-significant-bit extraction $m\&(-m)$,
\item one table lookup for $\lookup{A}[f]$,
\item two bit tests, namely $(y\&f)=0$ and $(p\&f)\neq 0$,
\item three XOR updates, namely $y\gets y\xor f$, $x\gets x\xor \lookup{V}[f]$, and $p\gets p\xor \lookup{C}[f]$,
\item one assignment of the next amplitude.
\end{itemize}
All of these are $O(1)$ word operations. So the loop costs $O(2^k)$ time.
\end{enumerate}

Adding the three components gives
\[
O(k^2) + O(2^k) + O(2^n) = O(2^n + 2^k + k^2).
\]
The space bound is obtained similarly.
\end{proof}

\begin{corollary}
Algorithm~\ref{alg:main} is asymptotically optimal for the explicit output task of converting quadratic form stabilizer data to a full state vector.
\end{corollary}

\begin{proof}
Theorem~\ref{thm:complexity} gives an $O(2^n)$ upper bound. Conversely, the task of writing an output vector of length $2^n$ imposes a trivial lower bound of $\Omega(2^n)$. Therefore the running time is $\Theta(2^n)$.
\end{proof}

\begin{corollary}[From a stabilizer check matrix to a full state vector]
\label{cor:check-matrix}
If one combines the $O(n^3)$ reduction from a stabilizer check matrix to quadratic form data in \cite{de2025fast} with Algorithm~\ref{alg:main}, then a full stabilizer state vector can be materialized from a check matrix description in time $O(2^n)$.
\end{corollary}

\begin{remark}[Implication for stabilizer-state verification]
De Silva et al. also give an $O(2^n)$ routine for extracting quadratic form data from an explicit state vector \cite{de2025fast}. Their verifier then reconstructs the state with the $O(n2^n)$ generator and compares against the input, which leads to an $O(n2^n)$ worst-case verification bound. Replacing that generator by Algorithm~\ref{alg:main} lowers that verification pipeline to $O(2^n)$ on the same model.
\end{remark}

\begin{example}[Single-qubit phase state]
Consider
\[
\ket{\psi}=\frac{1}{\sqrt 2}\bigl(\ket{0}+\ii\ket{1}\bigr).
\]
This is a stabilizer state, specifically the $+1$ eigenstate of $Y$. A quadratic form description is
\[
n=1,\quad k=1,\quad h=0,\quad v_1=1,\quad d_1=1,\quad J=[0],\quad \gamma=\frac{1}{\sqrt 2}.
\]
Then $a_1=1$ and $B=[0]$, so the parity word is always $0$. The algorithm writes
\[
\psi_0=\gamma,\qquad \psi_1=\gamma\,\ii,
\]
which is the desired state vector.
\end{example}

\begin{example}[Two-qubit cluster state]
Consider the two-qubit cluster state
\[
\ket{C_2}=\frac{1}{2}\bigl(\ket{00}+\ket{01}+\ket{10}-\ket{11}\bigr).
\]
A quadratic form description is
\[
n=2,\quad k=2,\quad h=00,\quad v_1=10,\quad v_2=01,\quad d=(0,0),
\]
and
\[
J=
\begin{bmatrix}
0 & 1\\
0 & 0
\end{bmatrix},
\qquad
\gamma=\frac12.
\]
Then
\[
a_1=a_2=0,
\qquad
B=
\begin{bmatrix}
0 & 1\\
1 & 0
\end{bmatrix}.
\]
So the parity word $p(y)=By$ tracks whether the \emph{other} support bit is active.

Using the standard Gray sequence on two bits, the flip words are
\[
f_1=01,\qquad f_2=10,\qquad f_3=01.
\]
The algorithm evolves as follows (here the lower bit corresponds to $y_1$):
\begin{center}
\small
\begin{tabular}{@{}cccccc@{}}
\toprule
\textnormal{Step} $m$ & \textnormal{Flip} $f_m$ & \makecell{\textnormal{Current}\\\textnormal{support word} $x$} & \makecell{\textnormal{Current}\\\textnormal{parity word} $p$} & \textnormal{Phase exponent} $q$ & \makecell{\textnormal{Written}\\\textnormal{amplitude}} \\
\midrule
$0$ & -- & $00$ & $00$ & $0$ & $+\frac12$ \\
$1$ & $01$ & $10$ & $10$ & $0$ & $+\frac12$ \\
$2$ & $10$ & $11$ & $11$ & $2$ & $-\frac12$ \\
$3$ & $01$ & $01$ & $01$ & $0$ & $+\frac12$ \\
\bottomrule
\end{tabular}
\normalsize
\end{center}
Reading the array in basis order $(00,01,10,11)$ gives
\[
\frac12(1,1,1,-1),
\]
which is the state vector of $\ket{C_2}$.
\end{example}

\subsection{Generalization to odd prime qudit dimensions}
\label{subsec:qudit-materialization}

We now describe the direct qudit analog of the dense vector materialization algorithm for quadratic forms \cite{gottesman1998fault, hostens2005stabilizer, gross2006hudson, beaudrap2013linearized}.
Fix an odd prime qudit dimension \(\ell\). All support arithmetic is over \(\F_\ell\).
The phase arithmetic is carried out in $\F_\ell$. Define
\[
\zeta_\ell = \exp(2\pi \ii/\ell),
\qquad
\rho_\ell = 2^{-1}\in\F_\ell.
\]

We briefly recall the stabilizer convention for qudits.  The one-qudit
computational basis is indexed by
$
\{\ket{j}:j\in\F_\ell\}.
$
The generalized Pauli shift and phase operators are
\[
X_\ell\ket{j}=\ket{j+1},
\qquad
Z_\ell\ket{j}=\zeta_\ell^{\,j}\ket{j},
\]
where the addition \(j+1\) is in \(\F_\ell\).  Hence, for
\(a,b\in\F_\ell\),
\[
X_\ell^a\ket{j}=\ket{j+a},
\qquad
Z_\ell^b\ket{j}=\zeta_\ell^{\,bj}\ket{j}.
\]
For \(n\) qudits and vectors \(\alpha,\beta\in\F_\ell^n\), we write
\[
X(\alpha)=\bigotimes_{r=1}^n X_\ell^{\alpha_r},
\qquad
Z(\beta)=\bigotimes_{r=1}^n Z_\ell^{\beta_r}.
\]
The generalized \(n\)-qudit Pauli operators are, up to scalar powers of
\(\zeta_\ell\), the operators \(X(\alpha)Z(\beta)\).
A qudit stabilizer state is the unique common \(+1\) eigenvector of an
abelian subgroup of this generalized Pauli group generated by \(n\)
independent commuting Pauli operators, with the generator phases chosen so
that the desired state has eigenvalue \(+1\).

The quadratic form expression in Theorem~3.3 of \cite{labib2022stabilizer} implies the following equivalent definition.

\begin{definition}[Odd prime qudit quadratic form]
\label{def:qudit-qf}
Let \(0\le k\le n\).  A prime-dimensional qudit quadratic form description consists of:
\begin{itemize}[leftmargin=1.5em]
\item a shift vector \(h\in\F_\ell^n\),
\item linearly independent basis vectors \(v_1,\dots,v_k\in\F_\ell^n\),
\item a vector \(d\in\F_\ell^k\),
\item an upper-triangular matrix \(J\in\F_\ell^{k\times k}\),
\item a nonzero complex scalar \(\gamma\).
\end{itemize}
For \(y\in\F_\ell^k\), define
\begin{equation}
\phi_\ell(y)
=
\sum_{t=1}^k a_t y_t
+
\sum_{t=1}^k r_t\bigl(y_t^2-y_t\bigr)
+
\sum_{1\le i<j\le k} B_{ij}y_i y_j
\in \F_\ell,
\label{eq:qudit-phase}
\end{equation}
where the derived coefficients are
\begin{equation}
\label{eq:qudit-components}
a_t=d_t+\rho_\ell J_{tt},
\qquad
r_t=\rho_\ell J_{tt},
\end{equation}
and \(B\in\F_\ell^{k\times k}\) is the symmetric zero-diagonal matrix
\[
B_{ij}=
\begin{cases}
J_{ij}, & i<j,\\
0, & i=j,\\
J_{ji}, & i>j.
\end{cases}
\]
The data specify amplitudes
\begin{equation}
\psi_{h+\sum_{t=1}^k y_t v_t}
=\gamma\,\zeta_\ell^{\phi_\ell(y)}
\qquad
\textnormal{for every }y\in\F_\ell^k,
\label{eq:qudit-qf-description}
\end{equation}
and \(\psi_x=0\) for every computational basis label
\(x\notin h+\operatorname{span}\{v_1,\dots,v_k\}\).
For normalized stabilizer states one has \(|\gamma|=\ell^{-k/2}\), possibly including a global phase, but the materialization algorithm only uses \(\gamma\) as the scalar appearing in Eq.~\eqref{eq:qudit-qf-description}.
\end{definition}

\subsubsection{Packed qudit words and Gray transitions}

The same word-RAM model as in the qubit case is used, generalized to a fixed odd prime \(\ell\).  A vector in \(\F_\ell^m\), for \(m\le n\), is stored in a packed word representation.  Componentwise addition, subtraction, multiplication by a fixed element of \(\F_\ell\), reading one specified coordinate, and using a packed label \(x\in\F_\ell^n\) as an output address all take \(O(1)\) word operations.

\begin{figure}
    \centering
    \includegraphics[width=\linewidth]{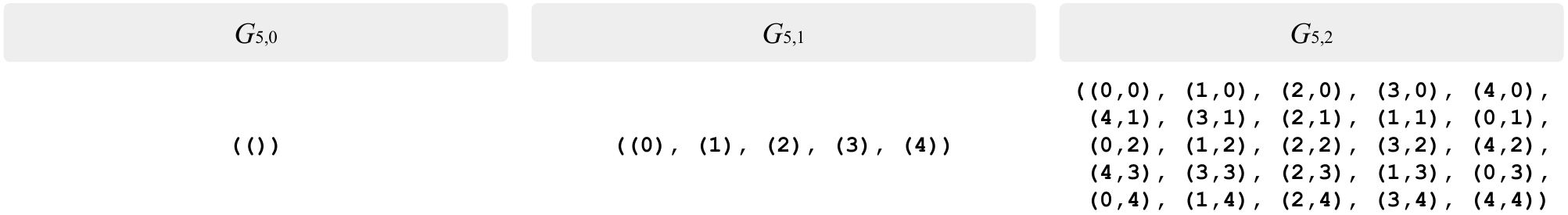}
    \caption{Visualization of the first three iterations in the construction of the reflected Gray list for $\ell=5$.}
    \label{fig:gray-code-example}
\end{figure}

We enumerate \(\F_\ell^k\) by the \(\ell\)-ary reflected Gray list \cite{mutze2022combinatorial}.  Let \(G_{\ell,0}\) be the one-element list containing the empty vector.  For \(k\ge1\), the list \(G_{\ell,k}\) is obtained by concatenating \(\ell\) blocks.  In the block with last coordinate \(a\in\{0,1,\dots,\ell-1\}\), the first \(k-1\) coordinates are listed in the order \(G_{\ell,k-1}\) if \(a\) is even and in the reverse order if \(a\) is odd (see Figure~\ref{fig:gray-code-example} for an example).
Write
\[
G_{\ell,k}=\bigl(y^{(0)},y^{(1)},\dots,y^{(\ell^k-1)}\bigr),
\qquad y^{(0)}=0.
\]
Consecutive elements differ in exactly one coordinate. For each
\(m=1,\dots,\ell^k-1\), there are a one-hot word \(f_m=e_{t_m}\in\F_\ell^k\) and a sign \(\delta_m\in\{+1,-1\}\) such that
\begin{equation}
 y^{(m)}=y^{(m-1)}+\delta_m f_m.
\label{eq:qudit-gray-transition}
\end{equation}
The sign \(\delta_m\) is interpreted in \(\F_\ell\). The transition list
\[
T_{\ell,k}=\bigl((f_m,\delta_m):m=1,\dots,\ell^k-1\bigr)
\]
can be precomputed in \(O(\ell^k)\) time and space for fixed \(\ell\).
For a one-hot word \(f=e_t\) and a packed vector \(u\in\F_\ell^k\), write
\[
\operatorname{coord}_f(u)=u_t.
\]

\subsubsection{The qudit phase update invariant}

As in Definition~\ref{def:interaction}, the off-diagonal part is stored in the symmetric zero-diagonal matrix \(B\).  For \(y\in\F_\ell^k\), define the cached interaction word
\[
p(y)=By\in\F_\ell^k.
\]
Thus
\[
p_t(y)=\sum_{j=1}^k B_{tj}y_j\in\F_\ell.
\]
This is the qudit analog of the parity word in the qubit algorithm.

\begin{lemma}[Qudit phase increment of a single Gray code flip]
\label{lem:qudit-phase-increment}
Let \(y\in\F_\ell^k\), let \(f=e_t\), and let \(\delta\in\{+1,-1\}\).  Then
\begin{equation}
\phi_\ell(y+\delta f)-\phi_\ell(y)
=
\delta a_t
+r_t\bigl(2\delta y_t+\delta^2-\delta\bigr)
+\delta\,p_t(y)
\quad\textnormal{in }\F_\ell.
\label{eq:qudit-phase-increment}
\end{equation}
\end{lemma}

\begin{proof}
The linear part contributes \(\delta a_t\).  The diagonal part contributes
\[
r_t\left((y_t+\delta)^2-(y_t+\delta)-y_t^2+y_t\right)
=r_t\bigl(2\delta y_t+\delta^2-\delta\bigr).
\]
The off-diagonal part contributes
\[
\delta
\left(\sum_{j<t}B_{jt}y_j+\sum_{t<j}B_{tj}y_j\right)
=\delta\,p_t(y),
\]
because \(B\) is symmetric and has zero diagonal.  Adding the three contributions gives Eq.~\eqref{eq:qudit-phase-increment}.
\end{proof}

\begin{lemma}
\label{lem:qudit-parity-update}
Let \(c_t\in\F_\ell^k\) be the \(t\)-th column of \(B\).  Then
\[
p(y+\delta e_t)=p(y)+\delta c_t
\]
for every \(y\in\F_\ell^k\) and every \(\delta\in\{+1,-1\}\).
\end{lemma}

\begin{proof}
Since \(p(y)=By\) over \(\F_\ell\),
\[
p(y+\delta e_t)=B(y+\delta e_t)=By+\delta Be_t=p(y)+\delta c_t.
\]
\end{proof}

\subsubsection{Qudit materialization algorithm}

The algorithm follows a similar structure to Algorithm~\ref{alg:main}.  It keeps the current Gray word \(y\), the current support label \(x\), the interaction word \(p=By\), and the current phase exponent \(q=\phi_\ell(y)\).  One-hot lookup tables are again used so that the algorithm never has to extract the index \(t\) from a one-hot word \(f=e_t\).  The tables are indexed by packed words in \(\F_\ell^k\), and only the one-hot locations are filled:
\[
\lookup{A}[e_t]=a_t,
\qquad
\lookup{R}[e_t]=r_t,
\qquad
\lookup{V}[e_t]=v_t,
\qquad
\lookup{C}[e_t]=c_t.
\]

\begin{algorithm}[t]
\caption{}
\label{alg:qudit-main}
\begin{algorithmic}[1]
\Require Odd prime qudit dimension \(\ell\), \(n\), \(k\), shift \(h\in\F_\ell^n\), linearly independent \(v_1,\dots,v_k\in\F_\ell^n\), phase data \(d\in\F_\ell^k\), upper-triangular \(J\in\F_\ell^{k\times k}\), scalar \(\gamma\in\C\setminus\{0\}\)
\Ensure Full array \((\psi_x)_{x\in\F_\ell^n}\)

\State Build \(a_1,\dots,a_k\), \(r_1,\dots,r_k\), and \(B\) from Eq.~\eqref{eq:qudit-components} and Definition~\ref{def:qudit-qf}
\State Build the reflected transition list \(T_{\ell,k}\)
\State Initialize one-hot lookup tables \(\lookup{A},\lookup{R},\lookup{V},\lookup{C}\) of length \(\ell^k\) to zero
\For{\(t=1\) to \(k\)}
    \State \(\lookup{A}[e_t]\gets a_t\)
    \State \(\lookup{R}[e_t]\gets r_t\)
    \State \(\lookup{V}[e_t]\gets v_t\)
    \State \(\lookup{C}[e_t]\gets\) the \(t\)-th column of \(B\)
\EndFor
\State Initialize a length-\(\ell^n\) complex array \(\psi\) to all zeros
\State \(y\gets 0\) \Comment{current Gray word in \(\F_\ell^k\)}
\State \(x\gets h\) \Comment{current support index in \(\F_\ell^n\)}
\State \(p\gets 0\) \Comment{current interaction word \(p(y)=By\) in \(\F_\ell^k\)}
\State \(q\gets 0\) \Comment{current phase exponent in \(\F_\ell\)}
\State \(\psi_x\gets \gamma\)
\For{each transition \((f,\delta)\) in \(T_{\ell,k}\), in order}
    \State \(u\gets \operatorname{coord}_f(y)\)
    \State \(w\gets \operatorname{coord}_f(p)\)
    \State \(q\gets q+\delta\lookup{A}[f]+\lookup{R}[f]\bigl(2\delta u+\delta^2-\delta\bigr)+\delta w\) in \(\F_\ell\)
    \State \(y\gets y+\delta f\)
    \State \(x\gets x+\delta\lookup{V}[f]\)
    \State \(p\gets p+\delta\lookup{C}[f]\)
    \State \(\psi_x\gets \gamma\,\zeta_\ell^q\)
\EndFor
\State \Return \(\psi\)
\end{algorithmic}
\end{algorithm}

\subsubsection{Correctness and complexity}

\begin{theorem}
\label{thm:qudit-correctness}
Algorithm~\ref{alg:qudit-main} returns the full amplitude vector of the qudit stabilizer state specified by the input quadratic form description.
\end{theorem}

\begin{proof}
Let
\[
G_{\ell,k}=\bigl(y^{(0)},y^{(1)},\dots,y^{(\ell^k-1)}\bigr)
\]
be the reflected Gray list, and write its transitions as in Eq.~\eqref{eq:qudit-gray-transition}.  For each \(m\), define
\[
x^{(m)}=h+\sum_{t=1}^k y_t^{(m)}v_t\in\F_\ell^n,
\qquad
q^{(m)}=\phi_\ell\bigl(y^{(m)}\bigr)\in \F_\ell,
\qquad
p^{(m)}=B y^{(m)}\in\F_\ell^k.
\]
We prove by induction on \(m\) that after the write corresponding to \(y^{(m)}\), the algorithmic variables satisfy
\[
y=y^{(m)},
\qquad
x=x^{(m)},
\qquad
p=p^{(m)},
\qquad
q=q^{(m)},
\]
and the array entry written at that moment is
\[
\psi_{x^{(m)}}=\gamma\,\zeta_\ell^{q^{(m)}}.
\]

\textbf{Base case $m=0$.} The reflected Gray list starts at \(y^{(0)}=0\).  The algorithm initializes \(y=0\), \(x=h\), \(p=0\), and \(q=0\), and writes \(\psi_h=\gamma\).  Since \(\phi_\ell(0)=0\) and \(B0=0\), the claim holds.

\textbf{Induction step.} Assume the claim holds after the \((m-1)\)-st write.  Let the next transition be \((f_m,\delta_m)\) with \(f_m=e_t\).  Then
\[
y^{(m)}=y^{(m-1)}+\delta_m e_t.
\]
By the induction hypothesis, the current variables before the update are \(y=y^{(m-1)}\), \(p=p^{(m-1)}\), and \(q=q^{(m-1)}\).  The algorithm reads
\[
u=y_t^{(m-1)},
\qquad
w=p_t\bigl(y^{(m-1)}\bigr).
\]
By Lemma~\ref{lem:qudit-phase-increment}, the update of \(q\) is
\[
q\gets q+
\phi_\ell\bigl(y^{(m)}\bigr)-
\phi_\ell\bigl(y^{(m-1)}\bigr),
\]
so the new phase exponent is \(q^{(m)}\).  Updating the support word gives
\[
y\gets y+\delta_m e_t=y^{(m)}.
\]
Since \(\lookup{V}[f_m]=v_t\), updating the support label gives
\[
x\gets x+\delta_m v_t
=h+\sum_{s=1}^k y_s^{(m)}v_s
=x^{(m)}.
\]
Since \(\lookup{C}[f_m]=c_t\), Lemma~\ref{lem:qudit-parity-update} gives
\[
p\gets p+\delta_m c_t
=p\bigl(y^{(m)}\bigr)
=p^{(m)}.
\]
The next write is therefore
\[
\psi_{x^{(m)}}=\gamma\,\zeta_\ell^{q^{(m)}},
\]
which is precisely Eq.~\eqref{eq:qudit-qf-description} at the support label indexed by \(y^{(m)}\).

The reflected Gray list visits every element of \(\F_\ell^k\) exactly once.  Because \(v_1,\dots,v_k\) are linearly independent over \(\F_\ell\), the map
\[
y\longmapsto h+\sum_{t=1}^k y_t v_t
\]
is injective.  Hence every support label is written exactly once.  All entries outside the affine support are initialized to zero and never written.  The returned array is therefore the correct full amplitude vector.
\end{proof}

\begin{theorem}
\label{thm:qudit-complexity}
For fixed odd prime \(\ell\), Algorithm~\ref{alg:qudit-main} runs in
$
O(\ell^n+\ell^k+k^2)
$
time and uses
$
O(\ell^n+\ell^k+k^2)
$
space on the packed word-RAM model.
\end{theorem}

\begin{proof}
Building \(a_1,\dots,a_k\), \(r_1,\dots,r_k\), and \(B\) from \(d\) and \(J\) costs \(O(k^2)\) time.  Constructing the reflected transition list costs \(O(\ell^k)\) time and space (this one-time precomputation can be excluded from the online complexity).  The one-hot lookup tables have length \(\ell^k\), and only \(k\) nonzero locations are filled, so their initialization costs \(O(\ell^k+k)\) time and \(O(\ell^k)\) space.  Initializing the dense output array costs \(O(\ell^n)\) time and \(O(\ell^n)\) space.

The main loop has \(\ell^k-1\) iterations.  Each iteration performs a constant number of packed coordinate reads, scalar operations in \(\F_\ell\), packed vector additions in \(\F_\ell^k\) or \(\F_\ell^n\), table lookups, and one output write.  By the packed word-RAM convention, each iteration costs \(O(1)\) word operations, so the loop costs \(O(\ell^k)\) time.

Adding the preprocessing, initialization, and loop costs gives
$
O(k^2)+O(\ell^k)+O(\ell^n)+O(\ell^k)
=O(\ell^n+\ell^k+k^2),
$
and the same accounting gives the space bound.
\end{proof}

\begin{corollary}
Algorithm~\ref{alg:qudit-main} is asymptotically optimal for the explicit output task of converting odd prime qudit quadratic form stabilizer data to a full state vector.
\end{corollary}

\begin{proof}
Theorem~\ref{thm:qudit-complexity} gives an \(O(\ell^n)\) upper bound.
Conversely, the task of writing an
output vector of length $\ell^n$ imposes a trivial lower bound of $\Omega(\ell^n)$.
Therefore the running time is \(\Theta(\ell^n)\).
\end{proof}

\section{Optimal materialization of Clifford gates}

One of the motivating applications listed in Section~\ref{sec:intro} was the construction of an explicit
Clifford matrix from a compact description. We now show that a similar Gray-walk approach
gives an optimal algorithm for that task as well. The first step is an optimal dense routine for
applying a Pauli operator to a state vector.

\subsection{Applying a Pauli operator to a dense vector}

For $u,w \in \F_2^n$, define
\begin{equation}
\label{eq:Xw_and_Zu_definition}
X(w) = X^{w_1} \otimes \cdots \otimes X^{w_n},
\qquad
Z(u) = Z^{u_1} \otimes \cdots \otimes Z^{u_n},
\end{equation}
where $X^0 = Z^0 = I$ and $X^1 = X$, $Z^1 = Z$. Every $n$-qubit Pauli operator can be written
in the form
\[
P = \omega X(w) Z(u),
\qquad
\omega \in \{\pm 1, \pm i\}.
\]
Also let $\epsilon_t \in \F_2^n$ denote the $n$-bit standard basis vector with a $1$ in coordinate
$t$ and $0$ elsewhere.

\begin{fact}
For every $x,u,w \in \F_2^n$,
\[
X(w) Z(u) |x\rangle = (-1)^{u^\top x} |x \oplus w\rangle.
\]
Consequently, if
\[
|\psi\rangle = \sum_{x \in \F_2^n} \psi_x |x\rangle,
\]
and $|\phi\rangle = P|\psi\rangle$, then
\begin{equation}
\label{eq:apply_pauli}
\phi_{x \oplus w} = \omega (-1)^{u^\top x} \psi_x
\qquad \textnormal{for every } x \in \F_2^n.
\end{equation}
\end{fact}

The next routine is the linear-phase analog of Algorithm~\ref{alg:main}. The support is now the whole cube $\F_2^n$, so the Gray walk runs over all basis labels. The only state invariant that must be cached is the parity $u^\top x$.

\begin{algorithm}[t]
\caption{}
\label{alg:pauli_to_vec}
\begin{algorithmic}[1]
\Require Dense array $(\psi_x)_{x \in \F_2^n}$, Pauli data $u,w \in \F_2^n$, $\omega \in \{\pm 1, \pm i\}$
\Ensure Dense array $(\phi_x)_{x \in \F_2^n}$ for $|\phi\rangle = \omega X(w) Z(u) |\psi\rangle$
\State Initialize a one-hot lookup table $\lookup{L}$ of length $2^n$ to zero
\For{$t = 1$ to $n$}
    \State $\lookup{L}[2^{t-1}] \leftarrow u_t$
\EndFor
\State Initialize a length-$2^n$ complex array $\phi$ to all zeros
\State $x \leftarrow 0$ \hfill $\triangleright$ current Gray word in $\F_2^n$
\State $z \leftarrow w$ \hfill $\triangleright$ current output index $z = x \oplus w$
\State $s \leftarrow 0$ \hfill $\triangleright$ current parity $u^\top x \in \F_2$
\State $\phi_z \leftarrow \omega \psi_x$
\For{$m = 1$ to $2^n - 1$}
    \State $f \leftarrow m\ \&\ (-m)$ \hfill $\triangleright$ one-hot Gray code flip word
    \State $x \leftarrow x \oplus f$
    \State $z \leftarrow z \oplus f$
    \State $s \leftarrow s \oplus \lookup{L}[f]$
    \If{$s = 0$}
        \State $\phi_z \leftarrow \omega \psi_x$
    \Else
        \State $\phi_z \leftarrow -\omega \psi_x$
    \EndIf
\EndFor
\State \textbf{return} $\phi$
\end{algorithmic}
\end{algorithm}

\begin{theorem}[Correctness and complexity of Algorithm~\ref{alg:pauli_to_vec}]
\label{thm:pauli_to_stab_correctness_and_complexity}
Algorithm~\ref{alg:pauli_to_vec} returns the dense vector of $P|\psi\rangle$ and runs in time $O(2^n)$ with space
$O(2^n)$. In particular, it is asymptotically optimal for this explicit output task.
\end{theorem}

\begin{proof}
Let $g(0),g(1),\dots,g(2^n - 1)$ be the binary reflected Gray code from Fact~\ref{fact:gray}. For each $m$,
define
\[
x^{(m)} = g(m) \in \F_2^n,
\qquad
z^{(m)} = x^{(m)} \oplus w \in \F_2^n,
\qquad
s^{(m)} = u^\top x^{(m)} \in \F_2.
\]
We prove by induction on $m$ that after the $m$-th write of the algorithm (counting the initial
write as $m=0$), the internal variables satisfy
\[
x = x^{(m)},
\qquad
z = z^{(m)},
\qquad
s = s^{(m)},
\]
and the array entry written at that moment is
\[
\phi_{z^{(m)}} = \omega (-1)^{s^{(m)}} \psi_{x^{(m)}}.
\]

\textbf{Base case $m=0$.} The algorithm sets $x=0$, $z=w$, and $s=0$, then writes
$\phi_w = \omega \psi_0$. Since $g(0)=0$, we have $x^{(0)} = 0$, $z^{(0)} = w$, and
$s^{(0)} = u^\top 0 = 0$. So the claim holds at $m=0$.

\textbf{Induction step.} Assume the claim holds after the $(m-1)$-st write, where
$1 \le m \le 2^n - 1$. By Fact~\ref{fact:gray},
\[
x^{(m)} = x^{(m-1)} \oplus \epsilon_t
\]
for the unique coordinate $t$ whose one-hot word is
\[
f_m = g(m) \oplus g(m-1) = m\ \&\ (-m).
\]
Because the current algorithmic variable is $x = x^{(m-1)}$, the update
\[
x \leftarrow x \oplus f_m
\]
sets $x$ to $x^{(m)}$. Since $z = x \oplus w$, the simultaneous update
\[
z \leftarrow z \oplus f_m
\]
sets $z$ to $z^{(m)}$.

It remains to check the parity bit. Only the $t$-th bit changes when one passes from
$x^{(m-1)}$ to $x^{(m)}$, so
\[
s^{(m)} = u^\top x^{(m)} \equiv u^\top x^{(m-1)} + u_t \pmod 2.
\]
By construction $\lookup{L}[f_m] = u_t$, and by the induction hypothesis the current variable is
$s = s^{(m-1)} = u^\top x^{(m-1)}$. Therefore the update
\[
s \leftarrow s \oplus \lookup{L}[f_m]
\]
sets $s$ to $s^{(m)}$.

The subsequent write is therefore
\[
\phi_{z^{(m)}} = \omega (-1)^{s^{(m)}} \psi_{x^{(m)}},
\]
which is the amplitude prescribed by Eq.~\eqref{eq:apply_pauli}.

Finally, Fact~\ref{fact:gray}(i) says that the Gray code visits every $x \in \F_2^n$ exactly once. The map
$x \mapsto x \oplus w$ is a bijection of $\F_2^n$, so every output position is written exactly once.
Hence the returned array is precisely the dense vector of $P|\psi\rangle$.

For the complexity, initializing the lookup table and the output array costs $O(2^n)$ time and
space. The Gray loop has $2^n - 1$ iterations. Each iteration performs a constant number
of word operations and one array write, so the loop costs $O(2^n)$ time. Thus the total running
time is $O(2^n)$ and the space usage is $O(2^n)$.
\end{proof}

\subsection{Dense Clifford matrix expansion}

A Clifford operator on $n$ qubits is a unitary $C$ such that $CPC^\dagger$ is a Pauli operator for
every $n$-qubit Pauli operator $P$ \cite{farinholt2014ideal}. For $1 \le t \le n$, define
\[
Z_t = I^{\otimes (t-1)} \otimes Z \otimes I^{\otimes (n-t)},
\qquad
X_t = I^{\otimes (t-1)} \otimes X \otimes I^{\otimes (n-t)}.
\]
The standard tableau of $C$ is the list of Pauli operators $U_t$ and $V_t$ defined in Eq.~\eqref{eq:Ut_and_Vt_definition}.
We assume that each $V_t$ is stored in the same binary form as above, namely by its phase and its
$X$- and $Z$-patterns.

\begin{lemma}[Gray-step recurrence for Clifford columns]
\label{lemma:column_recurrence}
Let $C$ be an $n$-qubit Clifford operator, and for each $x \in \F_2^n$ let
\[
c_x = C |x\rangle
\]
denote the $x$-th column of its full matrix. Then for every $x \in \F_2^n$ and every
$t \in \{1,\dots,n\}$,
\[
c_{x \oplus \epsilon_t} = V_t c_x.
\]
\end{lemma}

\begin{proof}
Because $X_t |x\rangle = |x \oplus \epsilon_t\rangle$, one has
\[
c_{x \oplus \epsilon_t} = C |x \oplus \epsilon_t\rangle = C X_t |x\rangle.
\]
Insert $I = C^\dagger C$ between $X_t$ and $|x\rangle$:
\[
C X_t |x\rangle = C X_t C^\dagger C |x\rangle = V_t c_x.
\]
\end{proof}

\begin{algorithm}[t]
\caption{}
\label{alg:cliff_mat}
\begin{algorithmic}[1]
\Require Clifford tableau $(U_1,V_1,\dots,U_n,V_n)$ and the first dense column $c_0 = C|0^n\rangle$
\Ensure Full $2^n \times 2^n$ complex matrix of $C$
\State Initialize a one-hot lookup table $\lookup{P}$ of length $2^n$ to zero
\For{$t = 1$ to $n$}
    \State $\lookup{P}[2^{t-1}] \leftarrow V_t$
\EndFor
\State Initialize a $2^n \times 2^n$ complex array $M$
\State $x \leftarrow 0$ \hfill $\triangleright$ current column label in $\F_2^n$
\State $c \leftarrow c_0$ \hfill $\triangleright$ current dense column $c = C|x\rangle$
\State Write $c$ as column $x$ of $M$
\For{$m = 1$ to $2^n - 1$}
    \State $f \leftarrow m\ \&\ (-m)$ \hfill $\triangleright$ one-hot Gray code flip word
    \State $x \leftarrow x \oplus f$
    \State $c \leftarrow$ the output of Algorithm~\ref{alg:pauli_to_vec} on input vector $c$ and Pauli data $\lookup{P}[f]$
    \State Write $c$ as column $x$ of $M$
\EndFor
\State \textbf{return} $M$
\end{algorithmic}
\end{algorithm}

\begin{theorem}
\label{thm:cliff_with_first_col}
Given the tableau of an $n$-qubit Clifford operator $C$ and its first dense column
$c_0 = C|0^n\rangle$, Algorithm~\ref{alg:cliff_mat} returns the full $2^n \times 2^n$ matrix of $C$ in time $O(4^n)$.
This bound is asymptotically optimal.
\end{theorem}

\begin{proof}
Let $g(0),g(1),\dots,g(2^n - 1)$ be the binary reflected Gray code from Fact~\ref{fact:gray}. For each $m$,
let
\[
x^{(m)} = g(m) \in \F_2^n,
\qquad
c^{(m)} = C|x^{(m)}\rangle.
\]
We prove by induction on $m$ that after the $m$-th write of the algorithm, the current variables
satisfy
\[
x = x^{(m)},
\qquad
c = c^{(m)},
\]
and the column written at that moment is the $x^{(m)}$-th column of the dense matrix of
$C$.

\textbf{Base case $m=0$.} The algorithm sets $x=0$ and $c=c_0 = C|0^n\rangle$, then writes $c$ as column
$0$ of $M$. Since $g(0)=0$, the claim holds.

\textbf{Induction step.} Assume the claim holds after the $(m-1)$-st write, where
$1 \le m \le 2^n - 1$. By Fact~\ref{fact:gray},
\[
x^{(m)} = x^{(m-1)} \oplus \epsilon_t
\]
for the unique coordinate $t$ whose one-hot word is
\[
f_m = g(m) \oplus g(m-1) = m\ \&\ (-m).
\]
Because the current algorithmic variable is $x = x^{(m-1)}$, the update
\[
x \leftarrow x \oplus f_m
\]
sets $x$ to $x^{(m)}$.

Also, by construction $\lookup{P}[f_m] = V_t$. By the induction hypothesis the current dense
column is $c = c^{(m-1)}$. Lemma~\ref{lemma:column_recurrence} therefore gives
\[
V_t c^{(m-1)} = c_{x^{(m-1)} \oplus \epsilon_t} = c_{x^{(m)}} = c^{(m)}.
\]
Algorithm~\ref{alg:pauli_to_vec} computes this new dense column in time $O(2^n)$, so after the update the variable
$c$ is $c^{(m)}$. The subsequent write therefore stores the correct column of the dense
matrix.

Fact~\ref{fact:gray}(i) says that the Gray code visits every $x \in \F_2^n$ exactly once. Hence every column
of the dense matrix is written exactly once, and the returned array is the full matrix of $C$.

For the runtime, initializing the one-hot lookup table costs $O(2^n)$ time. Initializing the output
matrix costs $O(4^n)$ time and space. The Gray loop has $2^n - 1$ iterations. In each iteration,
Algorithm~\ref{alg:pauli_to_vec} is invoked once, which costs $O(2^n)$ by Theorem~\ref{thm:pauli_to_stab_correctness_and_complexity}, and the resulting column is
written into the output matrix, which is also $O(2^n)$. Thus the total loop cost is
\[
(2^n - 1) \cdot O(2^n) = O(4^n).
\]
Including the initialization costs, the total running time is $O(4^n)$.
\end{proof}

\begin{corollary}[From a standard Clifford tableau to a full dense matrix]
\label{cor:cliff_materialization}
A full dense Clifford matrix can be materialized from a standard tableau description in time
$O(4^n)$.
\end{corollary}

\begin{proof}
Let $C$ be the Clifford operator represented by the tableau, and let
\[
c_0 = C|0^n\rangle
\]
be its first dense column. For each $t$,
\[
U_t c_0 = C Z_t C^\dagger C |0^n\rangle = C Z_t |0^n\rangle = C|0^n\rangle = c_0.
\]
Thus $c_0$ is stabilized by $U_1,\dots,U_n$. Because these operators are conjugates of the
independent commuting Hermitian Pauli operators $Z_1,\dots,Z_n$, they form a standard
stabilizer check matrix description of $c_0$. By Corollary~\ref{cor:check-matrix}, the vector $c_0$ can therefore be
materialized in time $O(2^n)$.

Applying Theorem~\ref{thm:cliff_with_first_col} to the tableau and the first column then yields the full dense matrix in
additional time $O(4^n)$, so the total complexity is $O(4^n)$.
\end{proof}

\begin{remark}[Implication for dense-Clifford verification]
De Silva et al.~\cite{de2025fast} also give an $O(n2^n)$ routine for extracting a Clifford tableau from a dense Clifford matrix. The verification pipeline reconstructs the dense matrix from the extracted tableau and compares against the input. Replacing that reconstruction step by
Corollary~\ref{cor:cliff_materialization} lowers the overall verification pipeline to $O(4^n)$ on the same word-RAM model.
\end{remark}

\section{Improved conversion from check matrices to quadratic forms}
\label{sec:fast-check-to-qf}

When the input is a stabilizer check matrix, the materialization pipeline has two stages:
\[
\text{check matrix}
\;\longrightarrow\;
\text{quadratic form}
\;\longrightarrow\;
\text{full state vector}.
\]
The second stage dominates because it has an unavoidable $\Omega(2^n)$ cost. The first stage, by contrast, is only polynomial-time preprocessing.
Consequently, even the previously known $O(n^3)$ reduction is asymptotically hidden by the dense materialization cost, and improving this polynomial step is negligible for the overall complexity.

Nevertheless, the compact-to-compact conversion is algorithmically meaningful in its own
right. In workflows that keep stabilizer states in compact representation and do not
materialize the full vector, the polynomial conversion cost is no longer negligible. For this reason, we also develop a faster reduction from check matrices to quadratic forms. Specifically, we replace the Gaussian elimination preprocessing by a \emph{sign-aware} Four Russians algorithm, while keeping the sign convention of Definition~\ref{def:check-matrix}.

\subsection{A sign convention compatible with the check matrix}

The check matrix in Definition~\ref{def:check-matrix} stores literal tensor products of $P(0,0)=I$, $P(1,0)=X$, $P(0,1)=Z$, and $P(1,1)=Y$.
For fast row operations it is convenient to use a \emph{canonical} binary Pauli row sign. For
$u,w\in\F_2^n$, define
\[
N(u,w)=\sum_{j=1}^n u_jw_j\in\Z,
\qquad
\langle u,w\rangle_2=N(u,w)\bmod 2\in\F_2,
\]
and
\[
\tau(u,w)=\frac{N(u,w)(N(u,w)+1)}{2}\bmod 2\in\F_2.
\]
Also define $X(w)$ and $Z(u)$ as in Eq.~\eqref{eq:Xw_and_Zu_definition}.
Then
\begin{equation}
\bigotimes_{j=1}^n P(w_j,u_j)
=
\ii^{N(u,w)}X(w)Z(u)
=
(-1)^{\tau(u,w)}(-\ii)^{\langle u,w\rangle_2}X(w)Z(u).
\label{eq:standard-to-canonical-pauli}
\end{equation}
Thus the generator row $(w\mid u)$ with standard sign bit $\sigma$ represents the same operator as the canonical row
\[
(w,u,\chi),
\qquad
\chi=\sigma\xor\tau(u,w),
\]
where
\begin{equation}
\mathcal P(w,u,\chi)=(-1)^\chi(-\ii)^{\langle u,w\rangle_2}X(w)Z(u).
\label{eq:canonical-pauli-row}
\end{equation}
This conversion only changes the bookkeeping notation, not the stabilizer generators themselves.

For every computational basis label $x\in\F_2^n$, the action of a canonical row is
\begin{equation}
\mathcal P(w,u,\chi)\ket{x}
=
(-1)^{\chi+\langle u,x\rangle_2}(-\ii)^{\langle u,w\rangle_2}\ket{x\xor w}.
\label{eq:canonical-row-action}
\end{equation}
All additions in exponents of $(-1)$ are in $\F_2$.

\begin{lemma}[Sign-aware product of commuting canonical rows]
\label{lem:safr-product}
Suppose the two canonical rows $(w,u,\chi)$ and $(w',u',\chi')$ commute, equivalently
\[
\langle u,w'\rangle_2\xor\langle u',w\rangle_2=0.
\]
Put
\[
a=\langle u,w\rangle_2,
\qquad
b=\langle u',w'\rangle_2,
\qquad
r=\langle u,w'\rangle_2.
\]
Then
\begin{equation}
\label{eq:canonical-sign-update}
\mathcal P(w,u,\chi)\mathcal P(w',u',\chi')
=
\mathcal P\bigl(w\xor w',\,u\xor u',\,\chi\xor\chi'\xor r\xor ab\bigr).
\end{equation}
\end{lemma}

\begin{proof}
Using $Z(u)X(w')=(-1)^{\langle u,w'\rangle_2}X(w')Z(u)$, we have
\[
X(w)Z(u)X(w')Z(u')
=
(-1)^rX(w\xor w')Z(u\xor u').
\]
The commutation assumption gives $\langle u',w\rangle_2=r$, and hence
\[
\langle u\xor u',w\xor w'\rangle_2
=a\xor b\xor r\xor r=a\xor b.
\]
For bits $a,b$,
\[
(-\ii)^a(-\ii)^b=(-1)^{ab}(-\ii)^{a\xor b}.
\]
Combining these identities gives the sign update in Eq.~\eqref{eq:canonical-sign-update}.
\end{proof}

\subsection{Sign-aware block extraction of the support directions}

The first task is to separate the rows whose $X$ parts span the support directions from the pure $Z$ constraints that cut out the affine support. Ordinary binary row addition is not enough because the sign bit is nonlinear under Pauli multiplication. Lemma~\ref{lem:safr-product} is the required sign update.

\begin{algorithm}[t]
\caption{}
\label{alg:safr-extract}
\begin{algorithmic}[1]
\Require Commuting independent canonical rows $(w_r,u_r,\chi_r)$, $1\le r\le n$, and a block size $b$
\Ensure Equivalent ordered rows whose first $k$ rows have independent $X$ parts and whose remaining rows have zero $X$ part
\State Partition the $n$ $X$-coordinates into consecutive blocks $J_1,J_2,\dots,J_t$, each of size at most $b$
\For{each $r=1,\dots,n$}
    \State Create a mutable row record $R_r\gets(w_r,u_r,\chi_r)$
\EndFor
\State $\mathcal A\gets[R_1,\dots,R_n]$ (the list of all active rows)
\State $\mathcal F\gets[]$ (the empty frozen list)
\For{each block $J_a$ in order}
    \State Read the current block pattern $w[J_a]$ of every active row $(w,u,\chi)\in\mathcal A$

    \State
    $
    W_a:=\operatorname{span}_{\mathbb F_2}\{w_R[J_a]:R\in\mathcal A\}
    \subseteq \mathbb F_2^{|J_a|}
    $
    \State Choose active rows $S_1,\dots,S_h\in\mathcal A$ such that
    $w_{S_1}[J_a],\dots,w_{S_h}[J_a]$ form a basis of $W_a$
    \State\Comment{$h=\dim W_a\le |J_a|\le b$}

    \For{$i=1$ to $h$}
        \State $b_i\gets w_{S_i}[J_a]\in\mathbb F_2^{|J_a|}$
    \EndFor
    \State $T[0]\gets(0,0,0)$
    \State $x\gets 0$
    \State $P\gets(0,0,0)$
    \For{$m=1$ to $2^h-1$}
        \State $f\gets m\mathbin{\&}(-m)$
        \State Let $i$ be the unique index such that $f=2^{i-1}$
        \State $x\gets x\oplus b_i$
        \State Replace $P$ by the canonical product $P\,S_i$ using Lemma~\ref{lem:safr-product}
        \State $T[x]\gets P$
    \EndFor

    \For{each active row $R\in\mathcal A$ not among $S_1,\dots,S_h$}
        \State $x\gets w_R[J_a]$
        \State Replace $R$ by the canonical product $R\,T[x]$ using Lemma~\ref{lem:safr-product}
    \EndFor
    \State Move $S_1,\dots,S_h$ from $\mathcal A$ to the end of $\mathcal F$
\EndFor
\State \Return the rows in $\mathcal F$, followed by the remaining rows in $\mathcal A$
\end{algorithmic}
\end{algorithm}

\begin{lemma}[Output structure of block extraction]
\label{lem:safr-extract-structure}
Algorithm~\ref{alg:safr-extract} returns rows of the form
\begin{equation}
(v_1,u_1,\chi_1),\dots,(v_k,u_k,\chi_k),(0,\rho_1,\zeta_1),\dots,(0,\rho_{n-k},\zeta_{n-k}),
\label{eq:safr-extracted-form}
\end{equation}
where $v_1,\dots,v_k$ are linearly independent and span the same subspace of $\F_2^n$ as the original $X$ parts. The generated stabilizer group is unchanged.
\end{lemma}

\begin{proof}
At the beginning of a block, every active row has zero $X$ entries in all earlier blocks. During the block, the selected rows have block patterns spanning all active block patterns. If an unselected active row has current pattern $x$, the table row $T[x]$ has the same current pattern. Multiplying by $T[x]$ therefore changes that block pattern to $x\xor x=0$. The selected rows also have zero $X$ entries in earlier blocks, so this operation cannot reintroduce a nonzero entry in an earlier block. By induction over the blocks, all rows still active at the end have zero $X$ part.

Every replacement is multiplication of one generator by a product of other commuting generators, with the sign computed by Lemma~\ref{lem:safr-product}. Such a replacement preserves the generated group, because the old generator can be recovered by multiplying the new generator by the same product of unchanged generators. Moving rows between the active and frozen lists only reorders the generators.

The frozen rows are the selected rows. In the block in which they are selected, their current block patterns are independent after all earlier active entries have been cleared. Looking at the earliest block that contributes to a putative linear dependence among frozen $X$ parts forces all coefficients of rows selected in that block to be zero. Repeating block by block proves independence. Since every nonselected active $X$ part is replaced by itself plus a combination of selected $X$ parts, the span of all $X$ parts is preserved throughout. At termination the active $X$ parts are zero, so the frozen $X$ parts span the original $X$-span.
\end{proof}

Let
\[
V=\operatorname{span}\{v_1,\dots,v_k\}\subseteq\F_2^n.
\]
Because all final rows commute, each pure-$Z$ vector $\rho_j$ is orthogonal to every $v_i$. Since row operations preserve independence of the full $2n$-bit row pairs, the vectors $\rho_1,\dots,\rho_{n-k}$ are independent. Therefore they form a basis of $V^\perp$.

\subsection{The support equations}

A pure-$Z$ row $(0,\rho_j,\zeta_j)$ acts on $\ket{x}$ by the scalar
$(-1)^{\zeta_j+\langle\rho_j,x\rangle_2}$. Hence every label $h$ in the support of the stabilizer state must satisfy
\begin{equation}
\langle\rho_j,h\rangle_2=\zeta_j,
\qquad
1\le j\le n-k.
\label{eq:safr-support-equations}
\end{equation}
Conversely, if $h$ satisfies Eq.~\eqref{eq:safr-support-equations}, then every vector in the affine space $h+V$ satisfies the same equations because $\rho_j\in V^\perp$.

The system in Eq.~\eqref{eq:safr-support-equations} is consistent for a valid check matrix. The final rows generate the same stabilizer group as the input rows, and a nonzero stabilized state has at least one computational basis label with nonzero amplitude. That label satisfies all pure-$Z$ equations. We find one solution $h$ by ordinary block Gaussian elimination on the augmented matrix with rows $(\rho_j\mid \zeta_j)$.

\subsection{Derivation of the phase formulas}

The final phase formulas must be written in the convention of Definition~\ref{def:qf}, where the exponent $d^\top y$ of $\ii$ is the ordinary integer sum of the selected $d_i$'s, interpreted modulo $4$. This convention differs slightly from a presentation in which the exponent is a single $\F_2$-linear form. The formulas below are the translated formulas in our notation.

For $y\in\F_2^k$, write
\[
z(y)=\sum_{j=1}^k y_jv_j\in V.
\]
A candidate output state has the form
\begin{equation}
\ket{\psi}
=
2^{-k/2}\sum_{y\in\F_2^k}\ii^{E(y)}\ket{h\xor z(y)},
\label{eq:safr-candidate-state}
\end{equation}
where
\begin{equation}
E(y)=\sum_{j=1}^k d_jy_j+2\sum_{1\le a\le b\le k}J_{ab}y_ay_b\pmod 4.
\label{eq:safr-E-def}
\end{equation}
This is Eq.~\eqref{eq:qf-description} with $\gamma=2^{-k/2}$.

Define the dot products
\begin{equation}
D_{ij}=\langle u_i,v_j\rangle_2,
\qquad
s_i=\langle u_i,h\rangle_2.
\label{eq:safr-D-and-s}
\end{equation}
Commutation of the top rows gives
\begin{equation}
D_{ij}=D_{ji}
\qquad(1\le i,j\le k),
\label{eq:safr-D-symmetric}
\end{equation}
because $\langle u_i,v_j\rangle_2\xor\langle u_j,v_i\rangle_2=0$.

Fix a top row $(v_i,u_i,\chi_i)$ and put $a_i=D_{ii}$. The coefficient of $\ket{h\xor z(y)}$ in $\mathcal P(v_i,u_i,\chi_i)\ket{\psi}$ comes from the summand indexed by $y\xor e_i$. By Eq.~\eqref{eq:canonical-row-action}, this coefficient equals
\[
2^{-k/2}\ii^{E(y\xor e_i)}
(-1)^{\chi_i+\langle u_i,h\xor z(y)\xor v_i\rangle_2}
(-\ii)^{a_i}.
\]

For the row to stabilize the candidate state, this must equal
$2^{-k/2}\ii^{E(y)}$ for every $y$. Equivalently,
\begin{equation}
E(y\xor e_i)-E(y)
\equiv
D_{ii}+2\left(\chi_i\xor s_i\xor D_{ii}\xor\sum_{j=1}^kD_{ij}y_j\right)
\pmod 4.
\label{eq:safr-phase-difference-required}
\end{equation}
Here the expression inside the parentheses is a bit in $\F_2$, while the outer factor $2$ is taken in $\Z_4$.

On the other hand, the definition in Eq.~\eqref{eq:safr-E-def} gives
\begin{equation}
E(y\xor e_i)-E(y)
\equiv
(1-2y_i)d_i+2J_{ii}
+2\sum_{j<i}J_{ji}y_j
+2\sum_{i<j}J_{ij}y_j
\pmod 4.
\label{eq:safr-phase-difference-from-J}
\end{equation}

Since $(1-2y_i)d_i\equiv d_i+2d_iy_i\pmod 4$, comparing Eqs.~\eqref{eq:safr-phase-difference-required} and \eqref{eq:safr-phase-difference-from-J} for all $y$ is achieved by setting
\begin{equation}
d_i=D_{ii},
\label{eq:safr-d-formula}
\end{equation}
\begin{equation}
J_{ii}=\chi_i\xor D_{ii}\xor s_i,
\label{eq:safr-Jdiag-formula}
\end{equation}
\begin{equation}
J_{ij}=D_{ij}
\qquad(1\le i<j\le k).
\label{eq:safr-Joff-formula}
\end{equation}
Indeed, the term $d_i=D_{ii}$ supplies the required $D_{ii}$ constant and the required $2D_{ii}y_i$ term, the diagonal entry $J_{ii}$ supplies the remaining constant $2(\chi_i\xor s_i\xor D_{ii})$, and the off-diagonal entries supply the terms $2D_{ij}y_j$ for $j\ne i$. Equation~\eqref{eq:safr-D-symmetric} ensures that the same upper-triangular entry works whether $j<i$ or $i<j$.

\begin{remark}
If one writes the phase instead as $(-1)^{Q(y)}\ii^{\ell(y)}$ with a single $\F_2$-linear value $\ell(y)=\sum_i D_{ii}y_i\bmod 2$, then the off-diagonal sign matrix contains the term $D_{ii}D_{jj}$. Definition~\ref{def:qf} uses $\ii^{d^\top y}$ with the ordinary integer sum $d^\top y$. Since
$
\ii^{\sum_i d_iy_i\bmod 2}
=
\ii^{\sum_i d_iy_i}(-1)^{\sum_{i<j}d_id_jy_iy_j},
$
that extra $D_{ii}D_{jj}$ term is absorbed by the integer $\ii$-exponent convention. Thus the compatible formula is the simpler off-diagonal rule in Eq.~\eqref{eq:safr-Joff-formula}.
\end{remark}

\subsection{The full conversion algorithm}

The remaining operations are standard block routines over $\F_2$.
\begin{itemize}[leftmargin=1.5em]
\item \textsc{BlockSolve} applies the Four-Russians block elimination technique \cite{albrecht2012m4rie} to the ordinary augmented linear system in Eq.~\eqref{eq:safr-support-equations}, with table entries storing XOR sums of binary rows rather than Pauli products.
\item \textsc{BlockedDots} computes all values in Eq.~\eqref{eq:safr-D-and-s} by partitioning the $n$ coordinates into blocks of size at most $b$. For each coordinate block $J$, it builds the table
$
T_J[x,t]=\langle x,a_t[J]\rangle_2
$
for every $x\in\F_2^{|J|}$ and every target $a_t\in\{v_1,\dots,v_k,h\}$, then XORs the table values for $x=u_i[J]$ into the running dot products.
\end{itemize}

\begin{algorithm}[t]
\caption{}
\label{alg:safr-sp2sq}
\begin{algorithmic}[1]
\Require A valid check matrix $(H,\sigma)$ as in Definition~\ref{def:check-matrix}
\Ensure Quadratic form $(h,v_1,\dots,v_k,d,J,\gamma)$ as in Definition~\ref{def:qf}
\State $b\gets \max\{1,\lfloor \log_2 n\rfloor-1\}$
\State $J \gets 0$
\For{$r=1$ to $n$}
    \State Read the $r$-th row of $H$ as $\left(w^{(r)}\mid u^{(r)}\right)$
    \State $N_r\gets \sum_{j=1}^n u^{(r)}_j w^{(r)}_j\in\Z$
    \State $\chi_r\gets \sigma_r\xor\left(\frac{N_r(N_r+1)}2\bmod 2\right)$
\EndFor
\State Run Algorithm~\ref{alg:safr-extract} on the canonical rows $\left(w^{(r)},u^{(r)},\chi_r\right)$ with block size $b$
\State\Comment{Let the output rows be ordered as in Eq.~\eqref{eq:safr-extracted-form}}
\State Use \textsc{BlockSolve} to find any $h\in\F_2^n$ satisfying $\langle\rho_j,h\rangle_2=\zeta_j$ for $1\le j\le n-k$
\State Use \textsc{BlockedDots} to compute $D_{ij}=\langle u_i,v_j\rangle_2$ and $s_i=\langle u_i,h\rangle_2$
\For{$i=1$ to $k$}
    \State $d_i\gets D_{ii}$
    \State $J_{ii}\gets \chi_i\xor D_{ii}\xor s_i$
\EndFor
\For{$1\le i<j\le k$}
    \State $J_{ij}\gets D_{ij}$
\EndFor
\State $\gamma\gets 2^{-k/2}$
\State \Return $(h,v_1,\dots,v_k,d,J,\gamma)$
\end{algorithmic}
\end{algorithm}

\subsection{Correctness and complexity}

\begin{theorem}
\label{thm:safr-correctness}
Algorithm~\ref{alg:safr-sp2sq} returns a quadratic form description of the stabilizer state specified by the input check matrix $(H,\sigma)$.
\end{theorem}

\begin{proof}
Equation~\eqref{eq:standard-to-canonical-pauli} shows that the initial conversion from $\sigma_r$ to $\chi_r$ preserves every input generator. Lemma~\ref{lem:safr-product} computes products of commuting canonical rows, and Lemma~\ref{lem:safr-extract-structure} shows that the extraction stage preserves the generated stabilizer group and outputs rows in the form in Eq.~\eqref{eq:safr-extracted-form}.

Let $V=\operatorname{span}\{v_1,\dots,v_k\}$. The pure-$Z$ vectors $\rho_1,\dots,\rho_{n-k}$ form a basis of $V^\perp$, and $h$ satisfies Eq.~\eqref{eq:safr-support-equations}. Therefore every pure-$Z$ row fixes every basis vector in the affine support $h+V$, and hence fixes the output state.

For each top row $(v_i,u_i,\chi_i)$, the derivation above shows that Eqs.~\eqref{eq:safr-d-formula}--\eqref{eq:safr-Joff-formula} make the required phase difference identity in Eq.~\eqref{eq:safr-phase-difference-required} hold for every $y$. Thus every top row stabilizes the output state. Consequently every final generator stabilizes the output state, and because the final generators generate the same stabilizer group as the input generators, every input generator stabilizes it as well.

The output state is normalized because it has $2^k$ nonzero amplitudes, all of magnitude $2^{-k/2}$. The input check matrix is valid, so its common $+1$ eigenspace is one-dimensional. Hence the normalized quadratic form state returned by the algorithm is the stabilizer state specified by $(H,\sigma)$.
\end{proof}

\begin{theorem}
\label{thm:safr-complexity}
On a valid $n$-qubit check matrix, Algorithm~\ref{alg:safr-sp2sq} runs in
$
O(n^3 / \log n)
$
bit operations and uses $O(n^2)$ bits of space.
\end{theorem}

\begin{proof}
For $n<4$ one may use any direct cubic reduction, so assume $n\ge4$ and let
$b=\lfloor\log_2 n\rfloor-1$. Then $b=\Theta(\log n)$ and $2^b\le n/2$.

A canonical row has $2n+1=O(n)$ bits. One sign-aware row product costs $O(n)$ bit operations. It XORs the two $X$ parts and the two $Z$ parts and computes the constant number of length-$n$ inner products appearing in Lemma~\ref{lem:safr-product}. In one coordinate block, selecting the current block basis costs $O(nb^2)$ bit operations, building the table of at most $2^b$ selected row products costs $O(2^bn)$, and clearing all unselected active rows costs $O(n^2)$. There are $O(n/b)$ blocks, so the extraction stage costs
$
O\!\left(\frac{n^3}{b}+\frac{n^2 2^b}{b}+n^2b\right).
$
The ordinary block solve for Eq.~\eqref{eq:safr-support-equations} has the same bound.

For the dot products, there are $O(n/b)$ coordinate blocks. For one block, the lookup table for the $k+1$ targets $v_1,\dots,v_k,h$ costs $O(k2^b)$ bit operations to build in Gray code order, and the table lookups for all pairs $(u_i,a_t)$ cost $O(k^2)$. Since $k\le n$, the total dot product cost is
$
O\!\left(\frac{k^2n}{b}+\frac{kn2^b}{b}\right)
\le
O\!\left(\frac{n^3}{b}+\frac{n^2 2^b}{b}\right).
$
The final formulas for $d$ and $J$ cost $O(k^2)\le O(n^2)$.

Substituting $b=\Theta(\log n)$ and $2^b\le n/2$ gives $O(n^3/\log n)$ bit operations overall. The matrix itself uses $O(n^2)$ bits. The largest block table in extraction or solving uses $O(2^bn)=O(n^2)$ bits, and the dot product table uses $O(2^bk)=O(n^2)$ bits. Thus the total space is $O(n^2)$ bits.
\end{proof}

\section{Discussion}

In many applications, a compact stabilizer or Clifford description is only the internal representation used to generate or manipulate an object, while the final interface requires explicit amplitudes or matrix entries.
The intrinsic cost of leaving the
stabilizer formalism is the size of the
dense object being requested, which entails \(2^n\) amplitudes for a state
and \(4^n\) entries for a unitary. We have shown that compact descriptions
of stabilizer states and Clifford transformations can be expanded at these rates, without any additional
polynomial overhead.
For a fixed odd prime qudit dimension, the same logic also gives optimal dense materialization of qudit stabilizer state vectors.

\paragraph{A structural interpretation.}
The off-diagonal part of the quadratic phase can be viewed as a graph on the $k$ support coordinates. The parity word $By$ is then the vector of neighborhood parities of the currently active vertex set $y$. Flipping one Gray code bit toggles exactly one vertex, which toggles the neighborhood parity of its adjacent vertices. That is why the whole update is one XOR with one adjacency column. In this form, the algorithm is conceptually simple, consisting of a Gray code traversal and the dynamic maintenance of neighborhood parity.

This output-optimal bridge is useful whenever stabilizer objects are used as
ingredients inside calculations that are not themselves purely
stabilizer calculations. In such workflows the
dense representation is often required only at the boundary of the
calculation. The results here show that, once that boundary is reached, the
conversion itself is \emph{not} an additional asymptotic bottleneck.

The broad applicability of our approach suggests two natural directions for future work:
\begin{itemize}
\item \textbf{Other explicit output conversions in the stabilizer formalism.} Whenever an existing routine uses Gray code to traverse the support and still scans a linear number of support coordinates per step, the same technique may remove that factor.
\item \textbf{Fermionic analogs.} Related compact phase representations also occur in Majorana stabilizer settings \cite{bravyi2005lagrangian, vijay2015majorana, litinski2018quantum, jiang2019majorana, beguvsic2026phase}. It is plausible that analogous word-parallel invariants exist there as well, although the details require separate analysis.
\end{itemize}

\section*{Acknowledgments}

This work is in part supported by the National Research Foundation of Korea (NRF, RS-2024-00451435 (20\%), RS-2024-00413957 (20\%)), Institute of Information \& communications Technology Planning \& Evaluation (IITP, RS-2025-02305453 (15\%), RS-2025-02273157 (15\%), RS-2025-25442149 (15\%), RS-2021-II211343 (15\%)) grant funded by the Ministry of Science and ICT (MSIT), Institute of New Media and Communications (INMAC), and the BK21 FOUR program of the Education, Artificial Intelligence Graduate School Program (Seoul National University), and Research Program for Future ICT Pioneers, Seoul National University in 2026.

\bibliographystyle{unsrt}
\bibliography{references}

@article{aaronson2004improved,
  title={Improved simulation of stabilizer circuits},
  author={Aaronson, Scott and Gottesman, Daniel},
  journal={Physical Review A},
  volume={70},
  number={5},
  pages={052328},
  year={2004},
  publisher={APS}
}

@article{dehaene2003clifford,
  title={Clifford group, stabilizer states, and linear and quadratic operations over {GF}(2)},
  author={Dehaene, Jeroen and De Moor, Bart},
  journal={Physical Review A},
  volume={68},
  number={4},
  pages={042318},
  year={2003},
  publisher={APS}
}

@article{gidney2021stim,
  title={Stim: a fast stabilizer circuit simulator},
  author={Gidney, Craig},
  journal={Quantum},
  volume={5},
  pages={497},
  year={2021},
  publisher={Verein zur F{\"o}rderung des Open Access Publizierens in den Quantenwissenschaften}
}

@article{de2025fast,
  title={Fast algorithms for classical specifications of stabiliser states and {C}lifford gates},
  author={de Silva, Nadish and Salmon, Wilfred and Yin, Ming},
  journal={Quantum},
  volume={9},
  pages={1586},
  year={2025},
  publisher={Verein zur F{\"o}rderung des Open Access Publizierens in den Quantenwissenschaften}
}

@article{albrecht2011efficient,
  title={{Efficient dense Gaussian elimination over the finite field with two elements}},
  author={Albrecht, Martin R and Bard, Gregory V and Pernet, Cl{\'e}ment},
  journal={arXiv preprint arXiv:1111.6549},
  year={2011}
}

@article{bard2006accelerating,
  title={{Accelerating Cryptanalysis with the Method of Four Russians}},
  author={Bard, Gregory V},
  journal={Cryptology ePrint Archive},
  year={2006}
}

@inproceedings{arlazarov1970economical,
  title={On economical construction of the transitive closure of a directed graph},
  author={Arlazarov, Vladimir L and Dinic, EA and Kronrod, MA and Faradzev, IA},
  booktitle={Dokl. Akad. Nauk SSSR},
  volume={194},
  number={11},
  pages={1209--1210},
  year={1970}
}

@article{albrecht2008efficient,
  title={{Efficient Multiplication of Dense Matrices over GF(2)}},
  author={Albrecht, Martin and Bard, Gregory and Hart, William},
  journal={arXiv preprint arXiv:0811.1714},
  year={2008}
}

@article{gusfield1997algorithms,
  title={{Algorithms on strings, trees, and sequences: Computer science and computational biology}},
  author={Gusfield, Dan},
  journal={ACM SIGACT News},
  volume={28},
  number={4},
  pages={41--60},
  year={1997},
  publisher={ACM New York, NY, USA}
}

@article{bard2008matrix,
  title={{Matrix inversion (or LUP-factorization) via the Method of Four Russians, in $\Theta(n^3/\log n)$ time}},
  author={Bard, Gregory},
  journal={LMS J. Comput. Math},
  volume={1},
  pages={14},
  year={2008}
}

@article{labib2022stabilizer,
  title={{Stabilizer rank and higher-order Fourier analysis}},
  author={Labib, Farrokh},
  journal={Quantum},
  volume={6},
  pages={645},
  year={2022},
  publisher={Verein zur F{\"o}rderung des Open Access Publizierens in den Quantenwissenschaften}
}

@article{raussendorf2001one,
  title={A one-way quantum computer},
  author={Raussendorf, Robert and Briegel, Hans J},
  journal={Physical Review Letters},
  volume={86},
  number={22},
  pages={5188},
  year={2001},
  publisher={APS}
}

@article{hein2004multiparty,
  title={Multiparty entanglement in graph states},
  author={Hein, Marc and Eisert, Jens and Briegel, Hans J},
  journal={Physical Review A},
  volume={69},
  number={6},
  pages={062311},
  year={2004},
  publisher={APS}
}

@article{van2004graphical,
  title={{Graphical description of the action of local Clifford transformations on graph states}},
  author={Van den Nest, Maarten and Dehaene, Jeroen and De Moor, Bart},
  journal={Physical Review A},
  volume={69},
  number={2},
  pages={022316},
  year={2004},
  publisher={APS}
}

@article{schlingemann2002stabilizer,
author = {Schlingemann, D.},
title = {Stabilizer codes can be realized as graph codes},
year = {2002},
publisher = {Rinton Press, Incorporated},
address = {Paramus, NJ},
volume = {2},
number = {4},
issn = {1533-7146},
journal = {Quantum Information and Computation},
pages = {307–323}
}

@article{schlingemann2001quantum,
  title={Quantum error-correcting codes associated with graphs},
  author={Schlingemann, Dirk and Werner, Reinhard F},
  journal={Physical Review A},
  volume={65},
  number={1},
  pages={012308},
  year={2001},
  publisher={APS}
}

@article{calderbank1998quantum,
  title={{Quantum error correction via codes over GF(4)}},
  author={Calderbank, A Robert and Rains, Eric M and Shor, Peter W and Sloane, Neil JA},
  journal={IEEE Transactions on Information Theory},
  volume={44},
  number={4},
  pages={1369--1387},
  year={1998},
  publisher={IEEE}
}

@article{shor1995scheme,
  title={Scheme for reducing decoherence in quantum computer memory},
  author={Shor, Peter W},
  journal={Physical Review A},
  volume={52},
  number={4},
  pages={R2493},
  year={1995},
  publisher={APS}
}

@article{steane1996error,
  title={Error correcting codes in quantum theory},
  author={Steane, Andrew M},
  journal={Physical Review Letters},
  volume={77},
  number={5},
  pages={793},
  year={1996},
  publisher={APS}
}

@article{cleve1997efficient,
  title={Efficient computations of encodings for quantum error correction},
  author={Cleve, Richard and Gottesman, Daniel},
  journal={Physical Review A},
  volume={56},
  number={1},
  pages={76},
  year={1997},
  publisher={APS}
}

@article{gottesman1998heisenberg,
  title={{The Heisenberg representation of quantum computers}},
  author={Gottesman, Daniel},
  journal={arXiv preprint quant-ph/9807006},
  year={1998}
}

@article{anders2006fast,
  title={Fast simulation of stabilizer circuits using a graph-state representation},
  author={Anders, Simon and Briegel, Hans J},
  journal={Physical Review A},
  volume={73},
  number={2},
  pages={022334},
  year={2006},
  publisher={APS}
}

@article{nest2010classical,
author = {Van den Nest, Maarten},
title = {{Classical simulation of quantum computation, the Gottesman-Knill theorem, and slightly beyond}},
year = {2010},
publisher = {Rinton Press, Incorporated},
address = {Paramus, NJ},
volume = {10},
number = {3},
issn = {1533-7146},
journal = {Quantum Information and Computation},
pages = {258–271}
}

@article{wietek2025xdiag,
  title={{XDiag: Exact diagonalization for quantum many-body systems}},
  author={Wietek, Alexander and Staszewski, Luke and Ulaga, Martin and Ebert, Paul L and Karlsson, Hannes and Sarkar, Siddhartha and Shackleton, Henry and Sinha, Aritra and Soares, Rafael D},
  journal={arXiv preprint arXiv:2505.02901},
  year={2025}
}

@article{gilchrist2005distance,
  title={{Distance measures to compare real and ideal quantum processes}},
  author={Gilchrist, Alexei and Langford, Nathan K and Nielsen, Michael A},
  journal={Physical Review A},
  volume={71},
  number={6},
  pages={062310},
  year={2005},
  publisher={APS}
}

@article{chuang1997prescription,
  title={{Prescription for experimental determination of the dynamics of a quantum black box}},
  author={Chuang, Isaac L and Nielsen, Michael A},
  journal={Journal of Modern Optics},
  volume={44},
  number={11-12},
  pages={2455--2467},
  year={1997},
  publisher={Taylor \& Francis}
}

@article{poyatos1997complete,
  title={{Complete characterization of a quantum process: the two-bit quantum gate}},
  author={Poyatos, JF and Cirac, J Ignacio and Zoller, Peter},
  journal={Physical Review Letters},
  volume={78},
  number={2},
  pages={390},
  year={1997},
  publisher={APS}
}

@article{nielsen2002simple,
  title={{A simple formula for the average gate fidelity of a quantum dynamical operation}},
  author={Nielsen, Michael A},
  journal={Physics Letters A},
  volume={303},
  number={4},
  pages={249--252},
  year={2002},
  publisher={Elsevier}
}

@article{greenbaum2015introduction,
  title={{Introduction to quantum gate set tomography}},
  author={Greenbaum, Daniel},
  journal={arXiv preprint arXiv:1509.02921},
  year={2015}
}

@article{bravyi2016improved,
  title={{Improved classical simulation of quantum circuits dominated by Clifford gates}},
  author={Bravyi, Sergey and Gosset, David},
  journal={Physical Review Letters},
  volume={116},
  number={25},
  pages={250501},
  year={2016},
  publisher={APS}
}

@article{bravyi2016trading,
  title={{Trading classical and quantum computational resources}},
  author={Bravyi, Sergey and Smith, Graeme and Smolin, John A},
  journal={Physical Review X},
  volume={6},
  number={2},
  pages={021043},
  year={2016},
  publisher={APS}
}

@article{bravyi2019simulation,
  title={{Simulation of quantum circuits by low-rank stabilizer decompositions}},
  author={Bravyi, Sergey and Browne, Dan and Calpin, Padraic and Campbell, Earl and Gosset, David and Howard, Mark},
  journal={Quantum},
  volume={3},
  pages={181},
  year={2019},
  publisher={Verein zur F{\"o}rderung des Open Access Publizierens in den Quantenwissenschaften}
}

@article{suzuki2021qulacs,
  title={{Qulacs: a fast and versatile quantum circuit simulator for research purpose}},
  author={Suzuki, Yasunari and Kawase, Yoshiaki and Masumura, Yuya and Hiraga, Yuria and Nakadai, Masahiro and Chen, Jiabao and Nakanishi, Ken M and Mitarai, Kosuke and Imai, Ryosuke and Tamiya, Shiro and others},
  journal={Quantum},
  volume={5},
  pages={559},
  year={2021},
  publisher={Verein zur F{\"o}rderung des Open Access Publizierens in den Quantenwissenschaften}
}

@article{smelyanskiy2016qhipster,
  title={{qHiPSTER: The quantum high performance software testing environment}},
  author={Smelyanskiy, Mikhail and Sawaya, Nicolas PD and Aspuru-Guzik, Al{\'a}n},
  journal={arXiv preprint arXiv:1601.07195},
  year={2016}
}

@article{guerreschi2020intel,
  title={{Intel Quantum Simulator: A cloud-ready high-performance simulator of quantum circuits}},
  author={Guerreschi, Gian Giacomo and Hogaboam, Justin and Baruffa, Fabio and Sawaya, Nicolas PD},
  journal={Quantum Science and Technology},
  volume={5},
  number={3},
  pages={034007},
  year={2020},
  publisher={IOP Publishing}
}

@article{magesan2011scalable,
  title={{Scalable and robust randomized benchmarking of quantum processes}},
  author={Magesan, Easwar and Gambetta, Jay M and Emerson, Joseph},
  journal={Physical Review Letters},
  volume={106},
  number={18},
  pages={180504},
  year={2011},
  publisher={APS}
}

@article{dankert2012exact,
  title={{Exact and approximate unitary 2-designs: constructions and applications}},
  author={Dankert, Christoph and Cleve, Richard and Emerson, Joseph and Livine, Etera},
  journal={arXiv preprint quant-ph/0606161},
  year={2012}
}

@article{emerson2007symmetrized,
  title={{Symmetrized characterization of noisy quantum processes}},
  author={Emerson, Joseph and Silva, Marcus and Moussa, Osama and Ryan, Colm and Laforest, Martin and Baugh, Jonathan and Cory, David G and Laflamme, Raymond},
  journal={Science},
  volume={317},
  number={5846},
  pages={1893--1896},
  year={2007},
  publisher={American Association for the Advancement of Science}
}

@inproceedings{hagerup1998sorting,
  title={{Sorting and searching on the word RAM}},
  author={Hagerup, Torben},
  booktitle={Annual Symposium on Theoretical Aspects of Computer Science},
  pages={366--398},
  year={1998},
  organization={Springer}
}

@article{koenig2014efficiently,
  title={{How to efficiently select an arbitrary Clifford group element}},
  author={Koenig, Robert and Smolin, John A},
  journal={Journal of Mathematical Physics},
  volume={55},
  number={12},
  year={2014},
  publisher={AIP Publishing}
}

@article{selinger2015generators,
  title={{Generators and relations for n-qubit Clifford operators}},
  author={Selinger, Peter},
  journal={Logical Methods in Computer Science},
  volume={11},
  year={2015},
  publisher={Episciences. org}
}

@article{bravyi2021hadamard,
  title={{Hadamard-free circuits expose the structure of the Clifford group}},
  author={Bravyi, Sergey and Maslov, Dmitri},
  journal={IEEE Transactions on Information Theory},
  volume={67},
  number={7},
  pages={4546--4563},
  year={2021},
  publisher={IEEE}
}

@article{savage1997survey,
  title={{A survey of combinatorial Gray codes}},
  author={Savage, Carla},
  journal={SIAM Review},
  volume={39},
  number={4},
  pages={605--629},
  year={1997},
  publisher={SIAM}
}

@article{gilbert1958gray,
  title={{Gray codes and paths on the $n$-cube}},
  author={Gilbert, Edgar N},
  journal={The Bell System Technical Journal},
  volume={37},
  number={3},
  pages={815--826},
  year={1958},
  publisher={Nokia Bell Labs}
}

@article{hostens2005stabilizer,
  title={{Stabilizer states and Clifford operations for systems of arbitrary dimensions and modular arithmetic}},
  author={Hostens, Erik and Dehaene, Jeroen and De Moor, Bart},
  journal={Physical Review A},
  volume={71},
  number={4},
  pages={042315},
  year={2005},
  publisher={APS}
}

@article{beaudrap2013linearized,
author = {de Beaudrap, Niel},
title = {{A linearized stabilizer formalism for systems of finite dimension}},
year = {2013},
publisher = {Rinton Press, Incorporated},
address = {Paramus, NJ},
volume = {13},
number = {1–2},
issn = {1533-7146},
journal = {Quantum Information and Computation},
pages = {73–115}
}

@inproceedings{gottesman1998fault,
  title={{Fault-tolerant quantum computation with higher-dimensional systems}},
  author={Gottesman, Daniel},
  booktitle={NASA International Conference on Quantum Computing and Quantum Communications},
  pages={302--313},
  year={1998},
  organization={Springer}
}

@article{gross2006hudson,
  title={{Hudson's theorem for finite-dimensional quantum systems}},
  author={Gross, David},
  journal={Journal of Mathematical Physics},
  volume={47},
  number={12},
  year={2006},
  publisher={AIP Publishing}
}

@article{mutze2022combinatorial,
  title={{Combinatorial Gray codes---an updated survey}},
  author={M{\"u}tze, Torsten},
  journal={arXiv preprint arXiv:2202.01280},
  year={2022}
}

@article{farinholt2014ideal,
  title={{An ideal characterization of the Clifford operators}},
  author={Farinholt, JM},
  journal={Journal of Physics A: Mathematical and Theoretical},
  volume={47},
  number={30},
  pages={305303},
  year={2014},
  publisher={IOP Publishing}
}

@inproceedings{albrecht2012m4rie,
  title={{The M4RIE library for dense linear algebra over small fields with even characteristic}},
  author={Albrecht, Martin R},
  booktitle={Proceedings of the 37th International Symposium on Symbolic and Algebraic Computation},
  pages={28--34},
  year={2012}
}

@article{beguvsic2026phase,
  title={{Phase-sensitive representation of Majorana stabilizer states}},
  author={Begu{\v{s}}i{\'c}, Tomislav and Chan, Garnet Kin},
  journal={arXiv preprint arXiv:2602.17604},
  year={2026}
}

@article{jiang2019majorana,
  title={{Majorana loop stabilizer codes for error mitigation in fermionic quantum simulations}},
  author={Jiang, Zhang and McClean, Jarrod and Babbush, Ryan and Neven, Hartmut},
  journal={Physical Review Applied},
  volume={12},
  number={6},
  pages={064041},
  year={2019},
  publisher={APS}
}

@article{vijay2015majorana,
  title={{Majorana fermion surface code for universal quantum computation}},
  author={Vijay, Sagar and Hsieh, Timothy H and Fu, Liang},
  journal={Physical Review X},
  volume={5},
  number={4},
  pages={041038},
  year={2015},
  publisher={APS}
}

@article{litinski2018quantum,
  title={{Quantum computing with Majorana fermion codes}},
  author={Litinski, Daniel and von Oppen, Felix},
  journal={Physical Review B},
  volume={97},
  number={20},
  pages={205404},
  year={2018},
  publisher={APS}
}

@article{bravyi2005lagrangian,
title = {{Lagrangian representation for fermionic linear optics}},
author = {Bravyi, Sergey},
year = {2005},
publisher = {Rinton Press, Incorporated},
address = {Paramus, NJ},
volume = {5},
number = {3},
issn = {1533-7146},
journal = {Quantum Information and Computation},
pages = {216–238}
}

@article{javadi2024quantum,
  title={{Quantum computing with Qiskit}},
  author={Javadi-Abhari, Ali and Treinish, Matthew and Krsulich, Kevin and Wood, Christopher J and Lishman, Jake and Gacon, Julien and Martiel, Simon and Nation, Paul D and Bishop, Lev S and Cross, Andrew W and others},
  journal={arXiv preprint arXiv:2405.08810},
  year={2024}
}

\end{document}